\documentclass[11pt,onecolumn]{article}
\usepackage{times}
\usepackage{amsmath}
\usepackage{amssymb}
\usepackage{xspace}
\usepackage{theorem}
\usepackage{graphicx}
\usepackage{epsfig}
\usepackage{ifpdf}
\usepackage{url,hyperref}
\usepackage{latexsym}
\usepackage{euscript}
\usepackage{xspace}
\usepackage{color}
\usepackage{makeidx}
\usepackage{picins,wrapfig}
\usepackage{stackrel}
\usepackage{algorithmic, algorithm}
\usepackage{amscd}
\usepackage[all]{xy}
\usepackage{lineno}

\long\def\remove#1{}

\newtheorem{theorem}{Theorem}[section] 
\newtheorem{lemma}[theorem]{Lemma}

\newtheorem{proposition}[theorem]{Proposition}
\newtheorem{definition}[theorem]{Definition}
\newtheorem{assumption}[theorem]{Assumption}
\newtheorem{newremark}[theorem]{Remark}
\newenvironment{proof}{{\em Proof:}}{\hfill{\hfill\rule{2mm}{2mm}}}

\definecolor{darkred}{rgb}{1, 0.1, 0.3}
\definecolor{darkgreen}{rgb}{0.5, 0.8, 0.1}
\definecolor{darkpurple}{rgb}{1.0, 0, 1.0}
\definecolor{darkblue}{rgb}{0, 0, 1.0}

\newcommand {\mm}[1] {\ifmmode{#1}\else{\mbox{\(#1\)}}\fi}

\newcommand{\eps}{{\varepsilon}}

\newcommand{\vv}		{\tilde{\nor}}
\newcommand{\nor}		{\mathbf{N}}

\newcommand{\inD}		{s}

\newcommand{\bb}		{\beta}

\newcommand{\conv}		{\mm {\rm Conv\,}}

\newcommand{\thetaM}{M_{\theta}}

\newcommand{\homo}	{{\sf H}}

\DeclareMathOperator{\argmin} {\mathrm argmin\,}

\newcommand{\manifold}		{{\mathsf{X}}}
\newcommand{\myflow}		{{\mathsf{F}}}
\newcommand{\lfs}			{{\mathrm{lfs}}}

\newcommand{\newtheta}		{{\tilde{\theta}}}
\newcommand{\cone}		{{c_2}}
\newcommand{\myqueue}	{{\mathbb{Q}}}
\newcommand{\lwfs}              {{\mathrm{wlfs}}}
\newcommand{\lnfs}              {{\mathrm{lnfs}_\beta}}
\newcommand{\lnfsp}             {{\mathrm{lnfs}}}
\newcommand{\aff}                       {{\mathrm{aff}}}

\begin{document}

\title{Parameter-free Topology Inference and Sparsification for Data on Manifolds}

\author{
Tamal K. Dey\thanks{
Department of Computer Science and Engineering,
The Ohio State University, Columbus, OH 43210, USA.
Email: {\tt \{tamaldey,dongzh,yusu\}@cse.ohio-state.edu}
}
\quad\quad
Zhe Dong$^*$
\quad\quad 
Yusu Wang$^*$
}

\date{}
\maketitle

\begin{abstract}
In topology inference from data, current approaches face two major problems.
One concerns the selection of a correct parameter to build 
an appropriate complex on top of the data points;
the other involves with the typical `large' size of this complex.
We address these two issues in the context of inferring homology
from sample points of a smooth manifold of known dimension 
sitting in an Euclidean
space $\mathbb{R}^k$. We show that, for a sample size of $n$ points,
we can identify a set of $O(n^2)$ points 
(as opposed to $O(n^{\lceil \frac{k}{2}\rceil})$ Voronoi vertices)
approximating a subset of the medial axis 
that suffices to compute a distance 
sandwiched between the well known {\em local
feature size} and the local {\em weak feature size}
(in fact, the approximating  set can be further reduced in size to $O(n)$).
This distance, called the {\em lean feature size},
helps pruning the input set at least to the
level of local feature size while making the data locally uniform.
The local uniformity in turn helps in
building a complex for homology inference
on top of the sparsified data without requiring any user-supplied 
distance threshold.
Unlike most topology inference results, ours does not require that the input
is dense relative to a {\em global} feature such as {\em reach}
or {\em weak feature size};
instead it can be adaptive with respect to the local feature size.
We present some empirical evidence in support of our
theoretical claims.
\end{abstract}

\section{Introduction}
\label{sec:intro}
In recent years, considerable progress has been made in
analyzing data for inferring the topology of a space from 
which the data is sampled. Often this process involves building
a complex on top of the data points, and then analyzing the
complex using various mathematical and computational tools
developed in computational topology. There are two main issues
that need attention to make this approach viable in practice. 
The first one stems from the 
requirement of choosing appropriate parameters
to build the complexes so that the provable guarantees align with the
computations.  The other one arises from the
unmanageable `size' of the complex---a problem
compounded by the fact that the input can be large and usual complexes
such as Vietoris-Rips built on top of it can be huge in size. 

In this paper, we address both of the above two issues
with a technique for data sparsification. 
The data points are assumed to be sampled from a smooth
manifold of known dimension sitting in some Euclidean space.
We sparsify the data so that 
the resulting set is locally uniform and is still good for homology inference.
Observe that, with a sample whose density varies with respect to a 
local feature size (such as the $\lfs$ proposed for surface 
reconstruction~\cite{AB98}), no global parameter for building
an appropriate complex can be found. The figure in the next paragraph illustrates this difficulty.

\parpic[r]{\includegraphics[height=2cm]{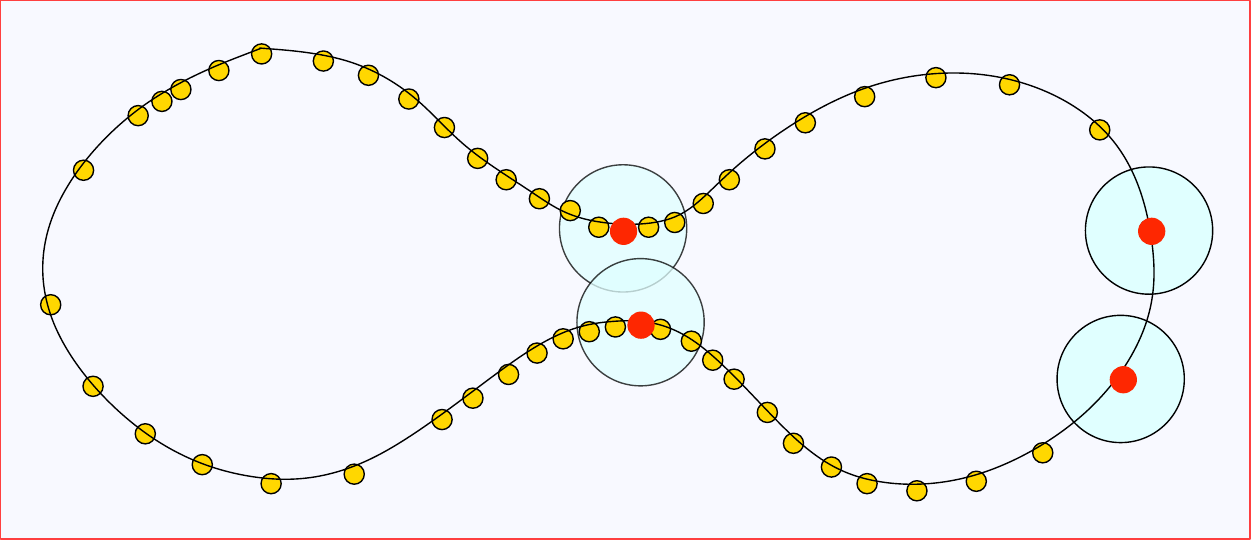}}
For the non-uniformly sampled curve, there is no single radius
that can be chosen to construct, for example, Rips or \v{C}ech complexes.
To connect points in the sparsely sampled part on right, the radius
needs to be bigger than the feature size at the
small neck in the middle. If chosen, this radius destroys
the neck in the middle thus creating spurious topology. 
Our solution to this problem is a sparsification strategy
so that the sample becomes locally uniform~\cite{CJL12,FR02} 
while guaranteeing that no topological information is lost. 
The sparsification is carried out without requiring any extra
parameter and the resulting local uniformity eventually
helps constructing the appropriate complex on top of the
sparsified set without  
requiring any user supplied parameter.

The sparsification also addresses the problem of `size'
because it produces a sub-sample of the original input.
The technique of subsampling 
has been suggested in some of the recent works. The well-known
witness complex builds on the idea of subsampling the input
data by restricting the Delaunay centers on the data points~\cite{SC04}.
Unfortunately, guarantees about topological inference cannot
be achieved with witness complexes unless some non-trivial 
modifications are made and parameters are tuned. 
Sparsified Rips complexes proposed by Sheehy~\cite{Sheehy} also
uses subsampling to
summarize the topological information contained in 
a Rips filtration (a nested sequence). 
The graph induced
complex proposed in~\cite{DFW13} alleviates the `size' problem
even further by replacing the Rips complexes with a more sparsified complex.
Both approaches, however, only approximate the true persistence
diagram and hence to infer homology exactly require a
user-supplied parameter to 
find the `sweet spot' in the filtration range. 
Furthermore, none of these
sparsifications is designed to work with
a non-uniform input that is adaptive to a {\em local} as opposed 
to a {\em global} feature size. 

Our algorithm first identifies a set of points
that supposedly approximates only a subset of the medial axis. 
It is known that the medial axis of a manifold embedded in 
$\mathbb{R}^k$ can be approximated with the Voronoi diagrams
of the $n$ input sample points~\cite{ACK01,CL05,Dey07} which requires 
$\Omega(n^{\lceil \frac{k}{2}\rceil})$ Voronoi vertices in the
worst-case. In contrast, we approximate the medial axis only 
with a {\em lean set} of $O(n^2)$ points
(which can be brought down to $O(n)$ with some more processing as shown 
in Section~\ref{appendix:sec:linearsize}). 
The distance to this lean set which we call the {\em lean feature size}
is shown to be sandwiched between the local feature size $\lfs$ and
the weak local feature size $\lwfs$.
Sparsifying the input with respect to this lean feature size
allows the data to be decimated at least to the level
of $\lfs$, but at the same time keeps it dense enough
with respect to the weak local feature size, which eventually leads to topological
fidelity. This roughly means that the
data is sparsified adaptively as much as possible without
sacrificing the topological information
(see experimental results in Figure~\ref{experiment}). 

The sparsified points are connected in a Rips-like complex
using the lean feature size computed for each sample point.
Following the approach in~\cite{CO08},
the guarantee for topological fidelity is obtained by interleaving
the union of a set of balls with the offsets of the manifold. 
To account for the adaptivity of the sample density, these
offsets are scaled appropriately by the
lean feature size and the approach in~\cite{CO08} is adapted
to this framework.
To the best of our knowledge, this is the first sparsification strategy
that handles adaptive input samples, produces an adaptive as well as a locally uniform
sparsified sample, and infers homology without requiring a threshold
parameter. 

 
\begin{figure*}[ht!]
\begin{center}
\begin{tabular}{c}
        \includegraphics[width=0.95\textwidth]{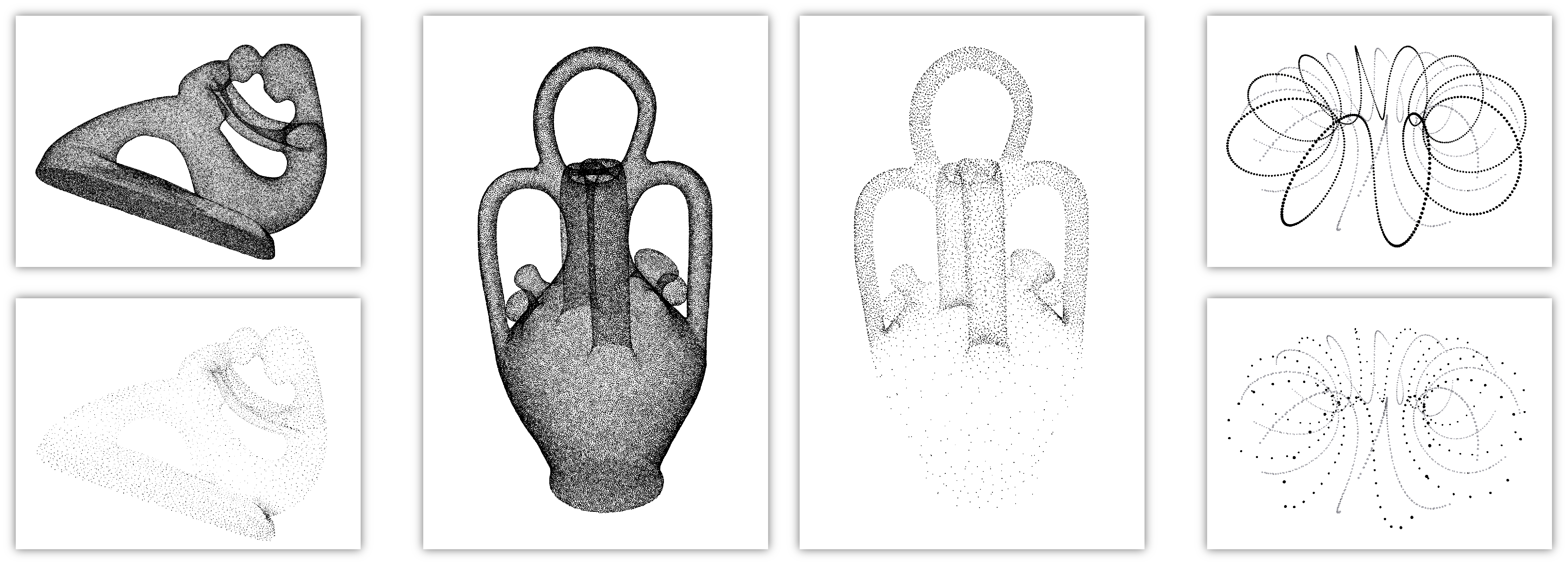}
\end{tabular}
\end{center}
\caption{Sparsification: Original samples of 126500 and
101529 points of 
{\sc MotherChild} and {\sc Botijo} are decimated to
6016 and 8622 points respectively. Betti numbers are computed correctly
by our algorithm (Section~\ref{sec:experiment}). The rightmost picture shows a 3D curve sample (top) 
and the {\it lean set} (bottom)
approximating a relevant subset of the medial axis which 
otherwise spans a much larger subspace of $\mathbb{R}^3$. 
}
\label{experiment}
\end{figure*}

\section{Sparsification}

Let $\manifold$ be a smooth compact manifold embedded in 
a $k$-dimensional ambient Euclidean space $\mathbb{R}^k$. 
Our goal is to 
sparsify a dense and possibly adaptive sample of $\manifold$ 
and still be able to recover homological information of 
$\manifold$ from it. 

\paragraph*{Distance function, feature size, and sample density.}
Let $d(x, A)$ denote the distance between a 
point $x\in \mathbb{R}^k$ and 
its closest point in a compact set $A\subset \mathbb{R}^k$.
Consider the \emph{distance function} $d_\manifold : \mathbb{R}^k \to \mathbb{R}$ defined as $d_\manifold(x) = d(x, \manifold)$. 
Let $\Pi(x) = \{ y \in \manifold \mid d(x, y) = d(x, \manifold)\}$ 
be the set of closest points of $x \in \mathbb{R}^k$ in $\manifold$. 
Notice that, 
for any $y \in \Pi(x)$, the segment $xy$ is contained in the normal space 
$\nor_x$ of 
$\manifold$ at $x$. 
The \emph{medial axis} $M$
of $\manifold$ is the closure of the set of points with at least two 
closest points in $\manifold$, and thus 
$M := {\mathrm closure\,}\{ m\in \mathbb{R}^k \mid |\Pi(m)| \ge 2 \}$.

The \emph{local feature size} at a point $x \in \manifold$, denoted by $\lfs(x)$, is defined as the smallest distance between $x$ and
the medial axis $M$; that is, $\lfs(x) = d(x, M)$~\cite{AB98}. 
There is another feature size definition that is particularly
useful for inferring homological information~\cite{CL07}. 
This feature size is defined as the distance to the critical
points of the distance function $d_\manifold$,
which is not differentiable everywhere. However, one can still define the following vector which extends the
concept of gradient to $d_\manifold$~\cite{Lieu04}.
Specifically, given any point 
$x \in \mathbb{R}^k \setminus \manifold$, let $c(x)$ be the center of 
the unique minimal enclosing ball $B_x$ enclosing $\Pi(x)$.
Define the \emph{gradient vector} at $x$: $\nabla_d(x)  = \frac{x - c(x)}{d(x, \manifold)}$ and the critical points 
$C:=\{x\in\mathbb{R}^k\,|\, \nabla_d(x)=0\}$.
The {\em weak local feature size}
at a point $x\in \manifold$, denoted by $\lwfs(x)$, is defined
as $\lwfs(x)=d(x,C)$.
Given an $\eps$-dense sample w.r.t.~the $\lfs$ which is known
as the $\eps$-sample in the literature~\cite{Dey07}, we would like to
sparsify it to a locally uniform sample 
w.r.t. some function, ideally $\lfs$, or $\lwfs$. 
This motivates the following definition.

\begin{definition}
A discrete sample $P\subset \manifold$ is called $c$-dense
w.r.t. a function 
$\phi:\manifold \rightarrow \mathbb{R}$ if $\forall x\in \manifold$,
$d(x, P)\leq c\cdot \phi(x)$. It is $c$-sparse if 
each pair of distinct points $p,q \in P$
satisfies $d(p,q)\geq c\cdot \phi(p)$. The sample $P$
is called $(c_1,c_2)$-uniform w.r.t. $\phi$ if
it is $c_1$-dense and $c_2$-sparse w.r.t. $\phi$.
\end{definition}

To produce a $(c_1,c_2)$-uniform sample
w.r.t. $\lfs$ or $\lwfs$
one needs to compute $\lfs$ or $\lwfs$ or their approximations.
This in turn needs the computation of at least a subset of the
medial axis or its approximation. One option is to
approximate this set using the Voronoi poles as in
~\cite{AB98,ACK01}. This proposition faces two difficulties.
First of all, it needs  
computing the Voronoi diagram in high dimensions. Second, 
approximating the medial axis may 
require a large number of samples when a manifold of a low co-dimension
is embedded in a high dimensional Euclidean space. 
To overcome this difficulty we propose
to compute a discrete set $L$ near $M$
of small cardinality which helps estimating
the distance to a subset of $M$
(See the curve sample in Figure \ref{experiment} for an example).
The set $L$ called the
{\em lean set} allows us to define an easily computable feature size
which we call {\em lean feature size}.
We show that this feature size is sandwiched between the $\lfs$
and $\lwfs$ thereby enabling us to sparsify an arbitrarily
dense sample to a $(c_1,c_2)$-uniform sample w.r.t. a function
bracketed by $\lfs$ and $\lwfs$. 
The constants $c_1,c_2$ are universal which ultimately 
leads to a parameter-free inference 
of the homology.

From now on, we assume that the input $P$ is a dense 
sample of $\manifold$ in the following adaptive sense~\cite{AB98}.
Each point is also equipped with a normal information as stated
in Assumption~\ref{An-assumption}.
We will see later how this normal information can be computed.
\begin{assumption}
The input point set $P$ is $\eps$-dense
w.r.t. $\lfs$ function on a compact smooth manifold 
$\manifold \subset \mathbb{R}^k$ of known dimension
without boundary. Also, every point
$p\in P$ has an estimated normal space $\vv_p$ where 
$\angle{(\vv_p,\nor_p)}\leq \nu_\eps=O(\eps)$ 
\footnote{We note that $\vv_p$ and $\nor_p$ here are subspaces of $\mathbb{R}^k$. The angle between them refers to the smallest non-zero \emph{principle angle} between these two subspaces as used in the literature. }
(see Section~\ref{sec:computation} for computations of $\vv_p$).
\label{An-assumption}
\end{assumption}

Notice that while we assume the input to be $\eps$-dense w.r.t. $\lfs$,
we do not need to know $\lfs$ and, locally, the sample can be much denser 
and non-uniform. 
Now we define the {\em lean set} with respect to which we define
the lean feature size.
\subsection{Lean set}
\label{sec:leanset}

\begin{definition}
A pair $(p,q)\in P\times P$ is
$\bb$-{\em good} for $0<\beta<\frac{\pi}{2}$ 
if the following two conditions hold:
\begin{enumerate}
\item $\max\{\angle{(\vv_p,pq)},\angle{(\vv_q,pq)}\} \leq \frac{\pi}{2}-\bb$.
\item Let $v=\frac{p+q}{2}$ be the midpoint of $pq$.
The ball $B(v,c_\bb d(p,q))$ does not contain any point of $P$ where
$c_\bb=\frac{1}{3}\tan\frac{\bb}{2}$. 
\end{enumerate}
\label{def:betagood}
\end{definition}
\begin{definition}
The {\em $\bb$-lean set $L_\beta$} is defined as: 
\[
L_\bb= \{ v|\mbox{ $v=\frac{p+q}{2}$ is the mid point of $pq$ where 
$(p,q)$ is a $\bb$-good pair\}.} 
\]
The $\bb$-lean feature size is defined as $\lnfs(x)=d(x,L_\bb)$.
\label{def:betaleanset}
\end{definition}


%

One of our main results is the following property of
the lean feature size ( recall the definition of $\nu_\eps$ in 
Assumption~\ref{An-assumption}). 

\begin{theorem}
Let $\theta,\bb$ be two positive constants
so that $\frac{\pi}{4}\geq \theta\geq \bb+\frac{3}{2}\sqrt\eps+\nu_\eps$ for
a sufficiently small $\eps\leq \frac{1}{8}\sin^2\theta$.
Then,
\begin{enumerate}
\item  
$
\lnfs(x) \le c_1\cdot \lwfs(x) 
$
for any point $x$ in $\manifold$,
\item 
$
\lnfs(p) \geq c_2\cdot \lfs(p)
$
for every point $p\in P$
\end{enumerate}
where
$c_1 = 1 + \cos \theta + \eps$,
$c_2=\frac{2c_0c_\beta}{1+c_0+2c_0c_\beta}$, and
$c_0=\sin(\beta-\nu_\eps)$, $c_\bb=\frac{1}{3}\tan\frac{\bb}{2}$
are positive constants.
\label{lnfs-thm}
\end{theorem}

The upper bound follows from 
Proposition~\ref{med-approx} which shows a stronger
result that
$\lnfs$ is bounded from above by the distance to
a subset of the medial axis characterized by an angle
condition. This set also contains
all critical points of the distance function
$d_\manifold$. First, we establish this result. 

\begin{definition}
The {\em $\theta$-medial axis} $\thetaM\subseteq M$ of $\manifold$ is defined as 
the set of points $m\in M$ 
where there exist two points 
$x,y\in \Pi(m)$ such that $\angle{xmy}\geq 2\theta$. 
\end{definition}
We will see later that the concept of
$\theta$-medial axis is also used as a bridge between
geometry and topology for our inference result. Our algorithm 
\emph{does not} approximate $M_\theta$, but rather, approximates
the distances to it by the the \emph{lean set}.

\begin{proposition}
Let $\theta,\bb$ be two positive constants
so that 
$\frac{\pi}{2}\geq \theta\geq \bb+\frac{3}{2}\sqrt\eps+\nu_\eps$ for
a sufficiently small $\eps\leq \frac{1}{8}\sin^2\theta$.
Let $x$ be any point in $\manifold$.
Then, 
$
\lnfs(x)=d(x, L_\bb) \le c\cdot d(x, M_\theta)
$
where
$c = 1 + \cos \theta + \eps$ is a positive constant.
\label{med-approx}
\end{proposition}
\begin{proof}
Let $m = \argmin d(x, M_\theta)$. By definition, we have
a pair of points $s,t$ in the manifold $\manifold$ so that the line segments $sm$ and
$tm$ subtends an angle larger than or equal to $2\theta$ and
both $sm$ and $tm$ are normal to $\manifold$ at $s$ and $t$ respectively.
Let $p\in P$ and $q\in P$ be the nearest sample points to $s$ 
and $t$ respectively. 
By the $\eps$-sampling condition of $P$, we have that $d(p, s) \le \eps \lfs(s)$ and thus $\angle{(\nor_s, \nor_p)} \le \eps$. 

In Appendix~\ref{appendix:A}, we show that the pair $(p, q)$ is $\beta$-good, hence its midpoint $\frac{p+q}{2}$ belongs to $L_\beta$. 
Notice that 
$\max \{ \lfs(s), \lfs(t)\} \le d(s, m) = d(t, m)$, and due to the $\eps$-sampling condition, $d(\frac{p+q}{2}, \frac{s+t}{2}) \le \eps d(s,m)$. We then have: 
\begin{align}
& d(\frac{p+q}{2}, m) \le d(\frac{s+t}{2}, m) + d(\frac{p+q}{2}, \frac{s+t}{2}) \le (\cos \theta + \eps) d(s, m); \nonumber \\
\Rightarrow~~ &d(x,L_\beta) \le d(x, \frac{p+q}{2}) \le d(x, m) + d(m, \frac{p+q}{2}) 
\leq d(x, m)+ d(s,m)(\cos\theta +\eps).
\label{eqn:A}
\end{align}
Since $s$ is a closest point of $m$ in $\manifold$, 
we have $d(s,m) = d(m, \manifold) \leq d(x,m)$. 
Combining this with Eqn (\ref{eqn:A}), it follows that
$$
d(x,L_\bb)\leq 
(1+\cos\theta+\eps) \cdot d(x, m).
$$
\end{proof}

We bound the distance $d(x,M_\theta)$ with $\lwfs(x)$ by observing
the following. 
The critical points of a distance function 
$d:\mathbb{R}^k\rightarrow \mathbb{R}$  
can be characterized by points $x\in \mathbb{R}^k$
that have the zero gradient $\nabla d$ along every
unit vector originating at $x$; see Grove~\cite{Grove}. 
It is also known that the critical points of the distance function 
$d_\manifold$ lie in the medial axis $M$. They are
points $m\in M$ so that the convex hull $\conv (\Pi(m))$
of all nearest neighbors of $m$ in $\manifold$ contains 
$m$. This means that there exists a pair of points $x,y$
in $\Pi(m)$ so that the angle $\angle{xmy}$ is large.
We use this angle condition to avoid the critical points.
Specifically, we show the following result for manifolds of arbitrary codimension which helps to 
make the angle condition precise. 

\begin{proposition}
Let the ambient dimension $k\geq 1$ and $m\in M$ be a critical point of the distance function 
$d_\manifold$. There exists a pair of points $x,y \in \Pi(m)$ so that
$\angle{xmy}\geq \frac{\pi}{2}$.
\label{large-angle}
\end{proposition} 
\begin{proof}
It is known that any critical point $m$ of the distance function
$d_\manifold$ is in
the convex hull $C=\conv \Pi(m)$ of the points in $\Pi(m)$.
This convex hull $C$ is a $j$-polytope for some $j\leq k$.
We can assume that $j$ is at least $2$,
because otherwise, $C$ is an edge with endpoints
say $x, y \in \Pi(m)$, and $\angle{xmy}=\pi \geq \frac{\pi}{2}$.

Now consider the subspace $\mathbb{R}^j \subseteq \mathbb{R}^k$ that
contains the $j$-polytope $C$.
Choose an arbitrary $2$-flat $H$ passing through $m$ in this $\mathbb{R}^j$.
The intersection of $H$ and $C$ is a polygon that contains $m$.
There is at least a pair of vertices $u$, $v$
of this polygon so that $\pi\ge\angle{umv}\geq \frac{\pi}{2}$.
The vertices $u$ and $v$ are the intersection of the $2$-flat with
the two codimension-2 faces $U$ and $V$ of $C$
respectively which are $(j-2)$-faces.

\parpic[r]{\includegraphics[height=2.5cm]{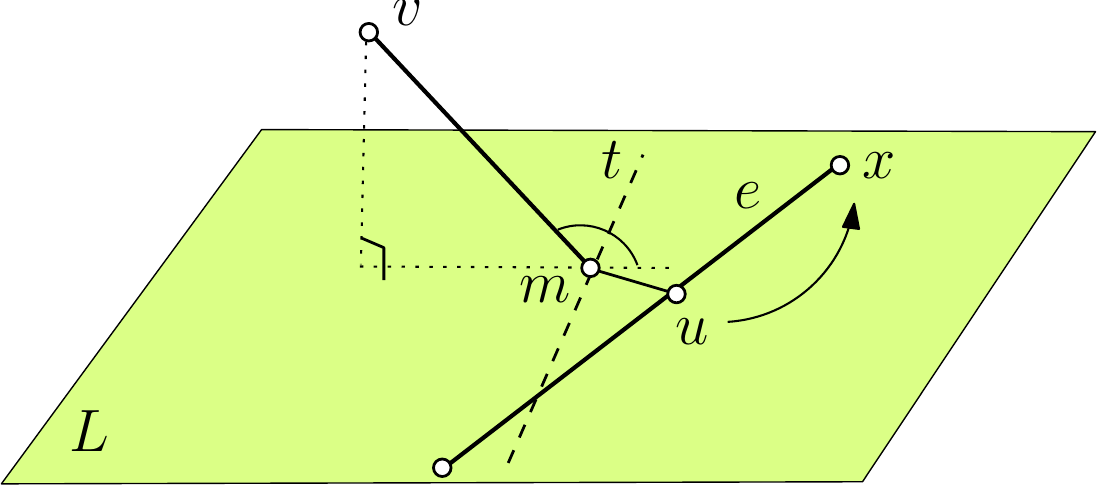}}
Let $e$ be the maximal line segment contained
in $U$ that connects $u$ and a vertex of $U$.
We can show that, one can choose an endpoint, say $x$, of $e$ so that
the angle $\angle{umv}$ remains at least $\frac{\pi}{2}$ when
$u$ assumes the position of $x$.
To see this consider the plane $L$ spanned by the line of $e$ and the
point $m$ (see figure on the right). Let $t$ be the line perpendicular
to the orthogonal projection of $mv$. Observe that all points $z\in e$
makes an angle $\angle{zmv}$ of at least $\frac{\pi}{2}$ if
$z$ lies in the halfplane of $L$ delimited by $t$ which does not contain
the projection of $mv$. Then, one of the endpoints of $e$ must satisfy
this condition because $u\in e$ does so to ensure
$\angle{umv}\geq \frac{\pi}{2}$.

The chosen endpoint $x$ of $e$ is  either
a vertex of $C$ or a point in a lower dimensional face of $U$.
Keeping $u$ at $x$, we can let $v$ coincide with a similar
endpoint of a line segment in $V$ while keeping the angle
$\angle{umv}$ at least $\frac{\pi}{2}$. Therefore, continuing this
process, $u$ and $v$ either reach a vertex of $C$ or a
lower dimensional face. It follows that both will reach a vertex
of $C$ eventually while keeping the angle
$\angle{umv}\ge \frac{\pi}{2}$. These two vertices qualify for $x$ and $y$
in the proposition.
\end{proof}

\begin{newremark}
We remark that the above bound of $\frac{\pi}{2}$
can be further tightened with a term
depending on the dimension $k$. However, the bound of $\frac{\pi}{2}$
suffices for our results. 
\end{newremark}
The following assertion is now immediate.
\begin{proposition}
For $\theta\leq \frac{\pi}{4}$, every point $x\in \manifold$ satisfies
$d(x,M_\theta)\leq \lwfs(x)$.
\label{crit-avoid}
\end{proposition}
Propositions~\ref{med-approx} and \ref{crit-avoid} together
proves the upper bound of the $\lnfs$ claimed in
Theorem~\ref{lnfs-thm}.
Next, we show the lower bound.

\begin{proposition}
For every sample point $p\in P$, we have
$\lnfs(p)> \cone \cdot \lfs(p)$ where
$\cone = \frac{2c_0 c_\bb}{1+c_0 + 2c_0 c_\bb}$ and $c_0=\sin(\beta-\nu_\eps)$.
\label{lb-to-med}
\end{proposition}
\begin{proof}
Let $z$ be the nearest point to $p$ in $L_\bb$, and $(p',q')$ the $\bb$-good pair that gives rise to $z$ (thus $z$ is the midpoint of $p'q'$). 
By definition of a $\beta$-good pair, 
$\angle{(\vv_{p'},p'q'))} \leq \frac{\pi}{2}-\beta$ and
hence $\angle{(\nor_{p'},p'q')}\leq \frac{\pi}{2}-\beta +\nu_\eps$.
There is a medial ball $B$ tangent to the manifold $\manifold$ at $p'$ so that the half line
$p'o$ going through the center $o$ of this ball $B$ realizes
the angle $\angle{(\nor_{p'},p'q')}$. Hence, 
$\angle{op'q'}\leq \frac{\pi}{2}-\beta +\nu_\eps$. 
It follows that
\begin{equation}
d(p',z) = \frac{1}{2}d(p',q')\geq d(p',o)\cos(\frac{\pi}{2}-\beta+\nu_\eps)
\geq c_0\cdot \lfs(p'), 
\mbox{ where $c_0=\sin(\beta-\nu_\eps)$.}
\label{eqn:B}\end{equation}
The empty ball condition of the $\bb$-good pair means that $2 c_\bb d(p', z) \le d(p, z)$, that is, $d(p', z) \le \frac{d(p,z)}{2c_\bb}$. It then follows that
$$d(p, p') \le d(p, z) + d(p', z) \le (1 + \frac{1}{2c_\bb}) d(p, z). $$
By the $1$-Lipschitz property of the $\lfs$ function and Eqn (\ref{eqn:B}), we have: 
\begin{align*}
\lfs(p) &\le \lfs(p') + d(p, p') \le \lfs(p') + (1 + \frac{1}{2c_\bb}) d(p, z) \le \frac{1}{c_0} d(p', z) + (1 + \frac{1}{2c_\bb}) d(p, z) \\
&\le \frac{1}{2c_0 c_\bb} d(p,z) + (1 + \frac{1}{2c_\bb}) d(p, z) = (1 + \frac{1}{2c_\bb} + \frac{1}{2c_0 c_\bb}) \cdot d(p,z). 
\end{align*}
Setting $\cone = \frac{1}{1 + \frac{1}{2c_\bb} + \frac{1}{2c_0 c_\bb}} = \frac{2c_0 c_\bb}{1 + c_0 + 2c_0 c_\bb}$, we have that $d(p,z) = \lnfs(p) \ge \cone \cdot \lfs(p)$, which proves the proposition. 
\end{proof}

We will see later that, $\beta$ is fixed at a constant value of $\frac{\pi}{5}$.
For this choice of $\beta$, $c_2$ is not unusually small.

\subsection{Computations for sparsification}
\label{sec:computation}

In this section we describe the algorithm {\sc Lean} 
that takes a standard $\eps$-dense sample $P$
w.r.t. $\lfs$ of a hidden 
manifold $\manifold \subset \mathbb{R}^k$ of known intrinsic
dimension, and outputs 
a \emph{sparsified set} $Q\subseteq P$. 
The set $Q$ is both adaptive and locally uniform as 
stated afterward in Theorem~\ref{thm:sparsification}. 
The parameter $\rho$ is chosen later to be a fixed constant
less than $1$.

\begin{algorithm}[h!]
\floatname{algorithm}{Algorithm}
\caption{{{\sc Lean}($P$,~$\beta$,~$\rho$)}}
\label{alg:lean}
\begin{algorithmic}[1]
        \STATE $L_\beta:=\emptyset$; 
        \FOR{every pair $(p,q)\in P\times P$}
                \STATE if $(p,q)$ is a $\beta$-good pair 
                then $L_\beta:=L_\beta \cup \{\frac{p+q}{2}\}$
	\ENDFOR
	\STATE Put $P$ in a max priority queue $\myqueue$ with 
priority $\lnfs(p)$
for $p\in P$;
	\WHILE{$\myqueue$ not empty }
		\STATE $q:=$extract-max($\myqueue$); ~~$Q:=Q\cup \{q\}$; 
                \STATE delete any $p$ from $\myqueue$ if $d(q,p)\leq \rho\lnfs(q)$ 
        \ENDWHILE
        \end{algorithmic}
\label{alg:Lean}
\end{algorithm}

The sparsification is based on the lean set $L_\bb$, which is computed in lines 2--4 of the algorithm.
We note that checking whether a pair $(p,q)$ is $\bb$-good or 
not requires 
no parameter other than $\bb$, which is set to a fixed constant
$\frac{\pi}{5}$ later in the homology inference algorithm. 
Clearly, $|L_\beta|=O(|P|^2)$ (see Section~\ref{appendix:sec:linearsize} for improving $|L_\beta|$
to $O(|P|)$).
There is one implementation detail which involves the
estimation of the normal space $\vv_p$ for every point $p \in P$. 
This estimation step is oblivious to 
any parameter but requires the
intrinsic dimension $s$ of $\manifold$ to be known. 

We estimate the tangent space $T_p$ 
(thus the normal space) of $\manifold$ at a point $p \in P$ as follows. 
Let $\inD$ be the intrinsic dimension of the manifold $\manifold$. 
Let $p_1 \in P$ be the nearest neighbor of $p$ in $P \setminus \{p\}$. 
Suppose we have already obtained points $\sigma_i = \{p, p_1, \ldots, p_i\}$ with $i < \inD$. 
Let $\aff(\sigma_i)$ denote the affine hull of the points in $\sigma_i$. 
Next, we choose $p_{i+1}\in P$ that is closest to $p$ 
among all points forming an angle within the range 
$[\frac{\pi}{2} - \frac{\pi}{5}, \frac{\pi}{2}]$ with $\aff(\sigma_i)$. 
We add $p_{i+1}$ to the set and obtain $\sigma_{i+1} = \{p, p_1, \ldots, p_i, p_{i+1}\}$. 
This process is repeated until $i+1 = \inD$, the dimension of $\manifold$,
at which point we have obtained $\inD+1$ points $\sigma_\inD = \{p, p_1, \ldots, p_{\inD}\}$. 
We use $\aff(\sigma_\inD)$ to approximate the tangent space $T_p$. 
It turns out that the simplex $\sigma_\inD$ obtained this way  has good thickness property, which by Corollary 2.6 in \cite{BG11} 
implies that the angle between the tangent space and the estimated 
tangent space at $p$ (thus also the angle between the normal space and 
the estimated normal space at $p$) is bounded by $O(\eps)$. 
The big-$O$ hides terms depending only on the intrinsic property 
of the manifold. 
See Appendix \ref{appendix:normal} for details. 
In other words, we have that the error $\nu_\eps$ in the estimated normal 
spaces (as required in Assumption \ref{An-assumption}) is $O(\eps)$.

Next, we put the points in $P$ in 
a priority queue and process them
in the non-decreasing order of their distances 
to $L_\beta$. We iteratively remove the point $q$ with
maximum value of $d(q,L_\beta)$ from the queue and proceed
as follows. We put $q$ into the sparse set $Q$ and
delete any point from the queue that lies at a distance
of at most $\rho\lnfs(q)$
from $q$. Since we consider points in non-decreasing
order of their distances to $L_\beta$, no earlier point
that is already in the sparse set $Q$ can be
deleted by this process. 

Determining if a pair $(p,q)$ is $\beta$-good takes $O(|P|)$ time.
This linear complexity is mainly due to the range queries for balls 
required for testing the `empty ball' condition $2$ for $\beta$-goodness. 
Therefore, for
$L_\beta=O(|P|^2)$, the algorithm spends $O(|P|^3)$ time in total.
This can be slightly improved to $O(|P|^{2-\frac{1}{k}}2^{O(\log^*|P|})$ using 
general spherical range query data structure in the ambient space $\mathbb{R}^k$~\cite{Agarwal}. 
Once the lean set is computed, the computation of $\lnfsp$
for all points involves computing the nearest
neighbor in $L_\beta$ for each point $p\in P$.
Using the method described in section~\ref{appendix:sec:linearsize}, we
can bring down the lean set size to $O(|P|)$. Then,
computing $\lnfs$ takes at most $O(|P|^2)$ time in total.
The actual sparsification in steps 6-9 takes only 
$O(|Q|^2)=O(|P|^2)$ time.

We show that the decimation by {\sc Lean} leaves the point
set $Q$ locally uniform w.r.t. $\lnfs$. 
The proof appears in 
Appendix~\ref{appendix:A}.
\begin{theorem}
Let $P$ be a sample 
of a manifold $\manifold \subseteq \mathbb{R}^k$, 
which is $\eps$-dense
w.r.t. $\lfs$.
For $\rho\le\frac{1}{12}$, the output of ${\sc Lean}(P,\beta,\rho)$ is a 
$(\frac{4}{3}\rho, \rho)$-uniform
sample of $\manifold$ w.r.t. $\lnfsp_{\beta}$ when $\eps>0$ 
is sufficiently small. 
\label{thm:sparsification}
\end{theorem}

\subsection{Linear-size Lean Set}
\label{appendix:sec:linearsize}

\newcommand{\reducedL} {{\widehat{L}}}

Observe that, the size $|L_\beta|$ is $O(n^2)$ if the input sample
$P$ has size $n$. This is far less than
$O(n^{\lceil\frac{k}{2}\rceil})$, $k$ being the ambient dimension,
which one incurs if the
medial axis is approximated with the Voronoi diagrams~\cite{CL05,Dey07}.
We can further thin
down the lean set to a linear size $O(n)$ for any fixed $k$ by the following simple strategy:

For every $p\in P$, among all $\bb$-good pairs $(p, q)$ it forms, we choose the pair $(p, q^*)$ such that the distance $d(p, q^*)$ is the smallest. We call this pair $(p, q^*)$ the \emph{minimal $\bb$-good pair} for $p$.
We now take a reduced lean set, denoted by $\reducedL_\beta$, as the collection of midpoints of these minimal $\bb$-good pairs. Obviously, $|\reducedL_\bb| = O(n)$.

Below we show that this reduced lean set can replace the original lean set $L_\bb$: it only worsens the distance from a sample point to the lean set by an additional constant factor. Note that this is the only distance in the end required by the algorithm (and the homology inference in Theorem \ref{thm:topoinfo-alg}). In particular, we have the following result.

\begin{lemma}
For any point $p \in P$, we have that $\lnfs(p) \le d(p, \reducedL_\bb) \le (1 + \frac{1}{c_\bb}) \lnfs(p).$
\label{lem:reducedLeanset}
\end{lemma}
\begin{proof}
The left inequality is trivial since $\reducedL_\bb \subseteq L_\bb$. We will show the right inequality. Fix any sample point $p \in P$, and let $m \in L_\bb$, the midpoint of a $\bb$-good pair $(s, t)$, be $p$'s nearest neighbor in the original lean set $L_\bb$.

Let $(s, t^*)$ be the minimal $\bb$-good pair for $s$, and $m^*$ its midpoint.
We now show that $d(p, \reducedL_\bb) \le d(p, m^*) \le (1 + \frac{1}{c_\bb}) d(p, L_\bb).$
Indeed, since $(s, t^*)$ is the minimal $\bb$-good pair for $s$, we have that $d(s, t) \ge d(s, t^*)$.
Hence

$$
d(m, m^*) \le d(m, s) + d(s, m^*) \le \frac{1}{2}(d(s, t) + d(s, t^*)) \le d(s, t).
$$
At the same time, by the empty-ball property of a $\bb$-good pair, we have that $d(p, m) \ge c_\bb d(s, t)$; that is, $d(s, t) \le \frac{1}{c_\bb} d(p, m)$.
Putting everything together, we obtain:
$$
d(p, \reducedL_\bb) \le d(p, m^*) \le d(p, m) + d(m, m^*)
\le d(p, m) + d(s, t) \le (1 +\frac{1}{c_\bb}) d(p, m)
= (1 + \frac{1}{c_\bb}) d(p, L_\bb).
$$
The claim then follows.
\end{proof}

\section{Homology inference}
\label{scaled-sec}
In this section, we aim to infer homology groups of a hidden manifold $\manifold$ from its point samples. Let $\homo_i(\cdot)$ denote the $i$-dimensional homology group. 
It refers to the singular homology when the argument is a manifold or a compact set, and to the simplicial homology when it is a simplicial complex. All homology groups in this paper are assumed to be defined over the finite field $\mathbb{Z}_2$. For details on homology groups, 
see e.g. \cite{Munkres}. 

The homology inference from a point sample of a hidden manifold
$\manifold$ has been researched extensively in 
the literature~\cite{CL06,CO08,DFW13,Sheehy}.
However, most of these work assume that the given sample 
$P\subset \manifold$
is globally dense, that is, $\eps$-dense w.r.t. to the {\em infimum}
of $\lfs$ or $\lwfs$. This strong assumption allows to infer
the homology from an appropriate 
offset of $P$ w.r.t. the distance $d(x,P)$, which
is represented with the union of balls of equal radii around
the sample points. As we indicated in the introduction, unfortunately,
when the sample is {\em adaptive} ($\eps$-dense w.r.t. a non-constant
function $\phi$),
there may not be such choice of a global radius so that the offset
captures the topology of $\manifold$.

To circumvent this problem, one needs to scale the distance
with the function $\phi$ that provides the adaptivity.
This idea was used in~\cite{CL06} where $\phi$ is taken as $\lfs$.
Approximating $\lfs$ is difficult, so we use $\lnfs$ instead
for scaling. Observe that the offset may intersect the medial
axis, but we argue that we can compute relevant offsets  
that never contains the critical points of the scaled distance,
thereby ensuring topological fidelity. 

\subsection{Scaled distance and its offsets}
In what follows we develop the results in more
generality by scaling the
distance $d_\manifold$ with the distance to a 
finite set $L\subset\mathbb{R}^k$. Later, in computations,
we replace $L$ by the lean set $L_{\frac{\pi}{5}}$ and the distance 
$d(x,L)$ with
$\lnfsp_{\frac{\pi}{5}}$ for $x\in\manifold$. 
Recall that $\Pi(x)$ denotes the set of closest neighbors of $x\in \mathbb{R}^k$ in $\manifold$. 

\begin{definition}
Given a finite set $L\subset \mathbb{R}^k$ such that $L \cap \manifold = \emptyset$,  
Let $h_L:\mathbb{R}^k \rightarrow \mathbb{R}$ be a scaled distance to the manifold
where
$$
h_L(x)= \frac{d(x,\manifold)}{d(x,\manifold)+d(x,L)}  = \frac{d(x,\Pi(x))}{d(x,\Pi(x))+d(x,L)}. 
$$
\end{definition}

We avoid the obvious choice of $h_L(x)=\frac{d(x,\manifold)}{d(x,L)}$ 
because that makes $h_L(x)$ unbounded at $L$.
We are interested in
analyzing the topology of the $\alpha$-offsets 
$\manifold_\alpha = h_L^{-1}[0,\alpha]$ of $h_L$
(clearly, $\manifold_0 = \manifold$ since $L \cap \manifold = \emptyset$) 
when $\manifold_\alpha\setminus \manifold$ 
does not include any critical points of $h_L$. 
This brings us to the
concept of flow induced by the distance
function which was studied in \cite{Grove} and later used in the 
context of sampling theory~\cite{CL05,GJ03,Lieu04}. 
The vector field $\nabla_{d_{\manifold}}$ as we defined
earlier is not continuous. However, as it is shown in \cite{Lieu04}, there exists a continuous flow $\myflow: \mathbb{R}^k \setminus \manifold \times \mathbb{R}^+ \rightarrow \mathbb{R}^k \setminus \manifold$ such that $\myflow(x, t) = x + \int_0^t \nabla_{d_{\manifold}}(\myflow(x, \tau)) d\tau$. 
For a point $x\in \mathbb{R}^k \setminus \manifold$, the image $\myflow(x, [0,t])$ of an
interval $[0,t]$ is called its 
\emph{flow line}. For a point $x\notin \manifold\cup M$, where $M$ is the medial axis of $\manifold$, the flowline 
$\myflow(x, [0,\infty])$ first coincides 
with the line segment $x\Pi(x)$ which is normal to the manifold $\manifold$. Once it reaches the medial axis $M$, it stays in $M$. 
We show that $h_L$ increases along the flow line of $d_\manifold$
in the $\alpha$-offset that we are interested in.
This, in turn, implies that the $\alpha$-offset of
our interest avoids the critical points of $h_L$.
\begin{proposition}
For $\theta\le \frac{\pi}{4}$, 
$\alpha< \frac{\cos 2\theta}{1+\cos 2\theta}$ and 
$\thetaM\cap \manifold_{\alpha}=\emptyset$,
the function $h_L$ 
increases along the flow line on the piece 
$\manifold_{\alpha}\cap \myflow(x,[0,\infty))$ where 
$x$ is any point in $\manifold_{\alpha}\setminus \manifold$. 
\label{monotone-prop}
\end{proposition}
\begin{proof}
First, observe that, due to Proposition~\ref{large-angle},
we can assert that $\manifold_\alpha\setminus\manifold$ 
contains no critical point
of $d_\manifold$ since $\manifold_\alpha\cap M_\theta=\emptyset$ and 
$\theta \le \frac{\pi}{4}$. Therefore, flow lines for every point
$x\in \manifold_\alpha\setminus \manifold$ are (topological) segments.
Consider an arbitrary point $y = \myflow(x, t)$ such that 
$y \in \manifold_\alpha$. Set $d = d(y, \manifold)$ and $\tilde{d} = d(y, L)$. 
Since $y \in \manifold_\alpha$, we have
\begin{equation}
h_L(y)\leq \alpha ~~\Longrightarrow~~ \frac{\tilde d}{d} \geq \frac{1-\alpha}{\alpha}.
\label{d1-eq}
\end{equation}

For arbitrary small $\Delta t > 0$, let $\Delta d$ and $\Delta \tilde d$
denote the changes in the distances $d$ and $\tilde d$ respectively when
we move on the flow line from 
$y = \myflow(x,t)$ to $y' = \myflow(x,t+\Delta t)$.
Observe that by the
triangle inequality, $\Delta \tilde d = |d(y, L) - d(y', L)| \leq d(y, y')$.
We claim that $\Delta d \ge  d(y, y') \cdot \cos 2\phi$ where 
$\phi$ is the maximum angle so that any point of 
$\manifold_\alpha \cap M$ belongs to $M_\phi$. 

\parpic[r]{\includegraphics[height=4cm]{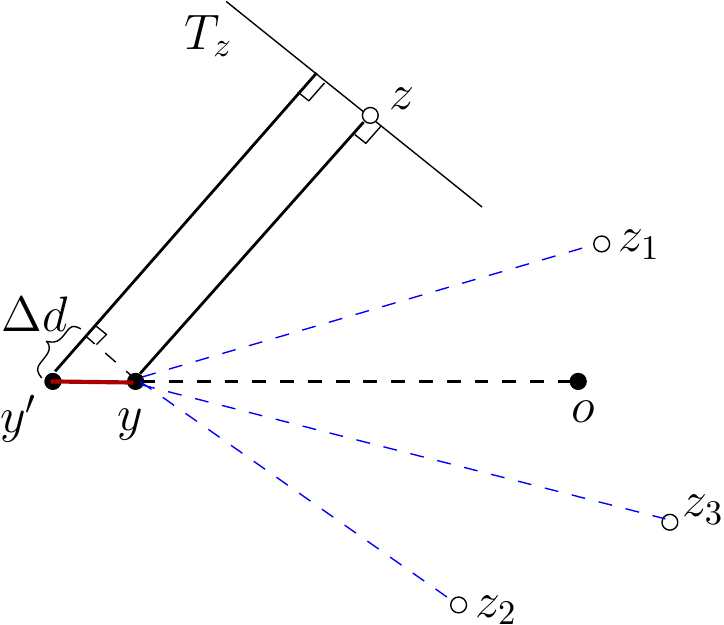}}
The flow line $\myflow(x, [0,\infty))$ follows 
a direction that is normal to the manifold $\manifold$ when 
it does not lie in the medial axis $M$ of $\manifold$. 
If $y$ lies on a portion of the flow line which is normal to the manifold $\manifold$, then it is easy to see that $\Delta d = |d(y, \manifold) - d(y', \manifold)| = d(y, y') \ge d(y, y') \cdot \cos 2\phi$. 
If $y$ lies on a portion of the flow line which is contained in the medial axis $M$, then the definition of $\phi$ implies that, 
for any two points $z_1, z_2 \in \Pi(y)$, 
the angle $\angle z_1 y z_2 \le 2 \phi$. 
At the same time, it is known that if $y \in M$, then the flow 
direction $\nabla_{d_{\manifold}} (y)$ at $y = \myflow(x, t)$ points in the 
direction of $\overrightarrow{oy}$ where $o$ is the center of the 
minimum enclosing ball for $\Pi(y)$ (see e.g, \cite{Lieu04}). 
In fact, $o$ must be contained in the convex hull of points in $\Pi(y)$. 
This further leads to that there exists a pair of points $z_1,z_2\in \Pi(y)$
so that the angle between $\overrightarrow{oy}$ and $\overrightarrow{zy}$ 
for any $z\in \Pi(y)$ is at most the angle $\angle z_1yz_2$, 
which is at most $2 \phi$. 
See the figure for an illustration where $\Pi(y) = \{z, z_1, z_2, z_3\}$, and $T_z$ is the intersection of the tangent space of $\manifold$ at $z$ with the plane spanned by $o, y, z$. 
Hence, in the limit as $y'\rightarrow y$, 
$\Delta d \rightarrow d(y, y') \cdot \cos \angle oyz$ for some $z \in \Pi(y)$, 
implying $\Delta d \ge d(y, y') \cdot \cos (2\phi)$. 

Finally, note that in the claim, we require that 
$M_\theta \cap \manifold_\alpha = \emptyset$. By definition of $\phi$,
this means that $\theta > \phi$. 
Hence, for $\theta\le\frac{\pi}{4}$, $\frac{\cos 2\theta}{1+\cos 2\theta} = 1 - \frac{1}{1+\cos 2\theta} \le 1 - \frac{1}{1+\cos 2\phi}$. The condition
$\alpha<\frac{\cos 2\theta}{1+\cos 2\theta}$ now provides that 
$\frac{1}{\cos 2\phi} < \frac{1-\alpha}{\alpha}$. 
It follows that: 
\begin{equation*}
\frac{\Delta \tilde d}{\Delta d} \leq \frac{1}{\cos 2 \phi}
< \frac{1-\alpha}{\alpha}\leq \frac{\tilde d}{d}\\
\Longrightarrow \frac{\tilde d +\Delta \tilde d}{d+\Delta d} < 
\frac{\tilde d}{d} \Longrightarrow h_L(\myflow(x,t)) < h_L(\myflow(x,t+\Delta t)).
\end{equation*}
\end{proof} 

Now, we will show that the $\alpha$-offset $\manifold_\alpha$
remains homotopy equivalent to $\manifold$ if $\alpha$ is
chosen appropriately. For the standard distance function
$d_\manifold$, such a result is well known~\cite{CL06,CL07}. Here, we
need the result for the scaled distance $h_L$ which we establish
using Proposition~\ref{monotone-prop} and the critical point 
theory of Grove~\cite{Grove}. 
The isotopy lemma of Grove~\cite{Grove} provides the partial
result that $\manifold_\alpha$ is homotopy equivalent to 
a smaller offset $\manifold_{\alpha'}$, $\alpha'<\alpha$.
Then we argue that $X_{\alpha'}$ is homotopy equivalent to $\manifold$
when $\alpha'$ is sufficiently small.

\begin{proposition}
Let $\theta \le \frac{\pi}{4}$ and 
$\alpha < \frac{\cos 2\theta}{1+\cos 2\theta}$.
Let $\manifold_{\alpha}$ be as defined in proposition~\ref{monotone-prop} where 
$\manifold_{\alpha}\cap \thetaM=\emptyset$. Then, 
$\manifold_\alpha$ is homotopy equivalent to $\manifold$
and hence $\homo_i(\manifold_\alpha)\cong \homo_i(\manifold)$ 
for each dimension $i\geq 0$. 
\label{manifold-homotopy}
\end{proposition}
\begin{proof}
Consider a real $\alpha'$ where $0<\alpha'\leq \alpha$.
Let $B={\mathrm Closure}\,(X_\alpha\setminus X_{\alpha}')$.
Any point $x\in B$ has a flow line $\myflow(x,[0,t])$ along which
$h_L$ strictly increases (Proposition~\ref{monotone-prop}). In particular,
there is a unit vector originating at $x$ along which
$\nabla_{h_L}$ does not vanish. Therefore, $B$ does not
contain any critical point of $h_L$. Applying the isotopy lemma
of Grove~\cite{Grove}, we conclude that $B$ deformation retracts 
to the bounding hypersurface $h_L^{-1}(\alpha')$ of $X_{\alpha}'$.
The resulting homotopy equivalence can be extended to
a map $r: \manifold_\alpha=B\cup \manifold_{\alpha'} 
\rightarrow h_L^{-1}(\alpha')\cup\manifold_{\alpha'}=\manifold_{\alpha'}$
by restricting $r$ to identity on $\manifold_{\alpha'}$.
It follows that $r$ is a homotopy equivalence. 

For any point $x\in \manifold_{\alpha'}\setminus \manifold$, 
a flow line $\myflow(x,[0,t])$ cannot re-enter 
$\manifold_{\alpha'}$ once it exits 
because of the monotonicity of $h_L$.
This means $\myflow(x,[0,t])$ intersects 
$\manifold_{\alpha'}$ in one connected
segment. 
Let $x'$ be the unique point where $\myflow(x,[0,t])$
intersects the hypersurface $h_L^{-1}(\alpha')$.
Since $\manifold$ is compact and smooth, by choosing $\alpha'>0$ sufficiently small, one can ensure that 
$\myflow(x,[0,t])\cap \manifold_{\alpha'}$ lies on the normal 
line segment $x \Pi(x)$, for all $x\in \manifold_{\alpha'}\setminus \manifold$. 
It implies that
$\manifold_{\alpha'}$ intersects the normal lines
to $\manifold$ in a connected segment along which
$\manifold_{\alpha'}$ can be retracted
to $\manifold$ completing the proof.
\end{proof}

\subsection{Interleaving and inference}
\label{sec:topology}
Our goal is to interleave the $\alpha$-offsets of $h_L$ with
the union of a set of balls $\cup B$ centered at the sample points 
because then, following the approach in~\cite{CO08}, we can relate
the topology of the nerve complex of $\cup B$ with that of $\manifold$.
For the distance function $d_\manifold$, the offsets restricted
to the sample $P$ provide the required set of balls because
$d_{\manifold}|_P$ approximates $d_\manifold$. Unfortunately, offsets
of $h_L$ restricted to $P$ are not necessarily union of geometric
balls centering points in $P$. 
Nevertheless, we show that a set of
balls whose 
radii are proportional to the distances to $L$
have the necessary property. 

First, we consider the union of balls, one for every point
in $\manifold$.  
Let $\cup B_{\alpha}$ denote the union of
balls $B(x,r)$ for every $x\in \manifold$ where $r=\alpha d(x, L)$.
One has the following interleaving result.

\begin{proposition}
$\manifold_{\frac{\alpha}{1+2\alpha}} \subseteq \cup B_{\alpha} \subseteq \manifold_{\alpha}$.
\label{interleave-prop1}
\end{proposition}
\begin{proof}
First we show the left inclusion. Let $x$ be any point in $\manifold_{\frac{\alpha}{1+2\alpha}}$, and $y$ an arbitrary point from $\Pi(x)$ (i.e, $d(x,y) = d(x, \Pi(x))$).
Then we have,
\begin{eqnarray*}
\frac{d(x,y)}{2d(x,y)+d(y,L)}
&=& \frac{d(x,y)}{d(x,y)+ (d(x,y)+d(y,L))}\\
&\leq& \frac{d(x,y)}{d(x,y)+ d(x,L)}
=\frac{\alpha}{1+2\alpha} \mbox{ since
$x\in \manifold_{\frac{\alpha}{1+2\alpha}}$}
\end{eqnarray*}
It then follows that
$$
(1+2\alpha)d(x,y)\leq 2\alpha d(x,y)+ \alpha d(y,L)
\Longrightarrow d(x,y) \leq \alpha d(y,L)
\Longrightarrow x \in \cup B_{\alpha}.
$$

We now prove the second inclusion. Let $x$ be any point in $\cup B_{\alpha}$.
Let $z\in \manifold$ be a point so that $x\in B(z, \alpha d(z,L))$;
that is, $d(x,z) \le \alpha d(z, L)$.
Such a point exists by the definition of $\cup B_{\alpha}$.
Using triangle inequality, we have:
\begin{eqnarray*}
h_L(x)= \frac{d(x,\manifold)}{d(x,\manifold)+d(x,L)}
\leq \frac{d(x,z)}{d(x,z)+d(x,L)} \leq \frac{d(x,z)}{d(z, L)} \leq \frac{\alpha d(z, L)}{d(z,L)} = \alpha.
\end{eqnarray*}
\end{proof}

We extend the above interleaving result to the union of balls whose
centers are restricted only to a sample $P\subset \manifold$. 
For convenience we define the following sampling condition closely related 
the $\eps$-dense sampling condition.
\begin{definition}
A finite set $P \subset \manifold$ is a $(\delta, L)$-sample of $\manifold$ if every point 
$x\in \manifold$ has a point $p\in P$ so that 
$d(x,p)\leq \delta d(p,L)$. 
Furthermore, let $\cup P_\alpha= \cup_{p\in P} B(p,\alpha d(p,L))$ 
denote the union of scaled balls around sample points in $P$. 
\end{definition}

\begin{newremark}
A $\delta$-dense sample w.r.t. $\lnfs$ is also a 
$(\frac{\delta}{1-\delta}, L_\bb)$-sample 
of $\manifold$. Conversely, a $(\delta, L_\bb)$-sample
of $X$ is also a $\frac{\delta}{1-\delta}$-dense sample 
w.r.t. $\lnfs$. These follow from the
fact that $\lnfs$ is $1$-Lipschitz.
\label{lnfs-remark}
\end{newremark}

\begin{proposition}
For a $(\delta,L)$-sample $P$ of $\manifold$ and any $\alpha >0$, we have
$\manifold_{\frac{\alpha}{1+2\alpha}} \subseteq \cup P_{\alpha+ \delta+\alpha\delta}  \subseteq \manifold_{\alpha+\delta+\alpha\delta}.$
\label{interleave-prop2}
\end{proposition}
\begin{proof}
Recall that by definition $\cup B_{\alpha}= \cup_{x\in \manifold} B(x, \alpha d(x,L))$.
By the $(\delta, L)$-sampling condition of $P$, as well as triangle inequality, we have $\cup B_\alpha \subseteq \cup P_{\alpha+\delta+\alpha\delta}$. 
Combining this with the left inclusion in Proposition ~\ref{interleave-prop1}, we have 
$$
\manifold_{\frac{\alpha}{1+2\alpha}} \subseteq \cup B_\alpha 
\subseteq \cup P_{\alpha+\delta+\alpha\delta}.
$$
The second inclusion follows because $\cup P_{\alpha+\delta+\alpha\delta} \subseteq \cup B_{\alpha+\delta+\alpha\delta}$ and $\cup B_{\alpha+\delta+\alpha\delta} \subseteq \manifold_{\alpha+\delta+\alpha\delta}$ (Proposition~\ref{interleave-prop1}).
\end{proof}
 
With the isomorphisms in the homology groups of the offset of our 
scaled distance function (Proposition \ref{manifold-homotopy}) and the 
interleaving result (Proposition \ref{interleave-prop2}), 
we can infer the homology of the hidden manifold $\manifold$ 
from the union of balls $\cup P_{\alpha}$. 

Suppose that $P$ is a $(\delta, L)$-sample of the manifold $\manifold$. 
Recall that $\cup P_{\alpha}$ denotes the union of 
balls $\bigcup_{p\in P} B(p, \alpha d(p, L) )$ centered at each 
point $p \in P$, with radius $\alpha d(p, L)$. 
Note that the parameter $\alpha$ does not stand for 
\emph{distance threshold}, but a \emph{scale parameter} for the 
distance $d(p, L)$. This parameter is universal for all points, 
while the distance $d(p, L)$ makes the union of balls adaptive. 

By manipulating the result in Proposition ~\ref{interleave-prop2}, 
one obtains that, for $\alpha + \delta \le \frac{1}{4}$ and 
$\alpha'=\frac{\alpha}{2(1-\alpha)}$, 
\begin{equation*}
\manifold_{\frac{\alpha}{2}}\subseteq \cup P_{\alpha'+\delta+\alpha'\delta}
\subseteq \cup P_{\frac{5}{4}\alpha'+\delta}
\subseteq \cup P_{\alpha+\delta} 
\end{equation*}
When $\alpha+\delta \le \frac{1}{6}$ and 
$\alpha'=\frac{\alpha+\delta}{1-2(\alpha+\delta)}$, similar manipulation gives
\begin{equation*}
\manifold_{\alpha+\delta}\subseteq \cup P_{\alpha'+\delta+\alpha'\delta}
\subseteq \cup P_{\frac{7}{6}\alpha'+\delta}
\subseteq \cup P_{\frac{11}{4}(\alpha+\delta)} 
\subseteq \cup P_{3(\alpha+\delta)}
\end{equation*}
So, for $\alpha+\delta \le \frac{1}{6}$, we obtain 
\begin{equation}
\manifold_{\frac{\alpha}{2}}\subseteq \cup P_{\alpha+\delta}
\subseteq \manifold_{\alpha+\delta} \subseteq \cup P_{3(\alpha+\delta)}
\subseteq \manifold_{3(\alpha+\delta)}
\label{eqn:C}
\end{equation}
which leads to inclusion-induced homomorphisms at the
homology level that interleave:
$$
\homo_i(\manifold_{\frac{\alpha}{2}})\rightarrow \homo_i(\cup P_{\alpha+\delta})
\rightarrow \homo_i(\manifold_{\alpha+\delta}) \rightarrow 
\homo_i(\cup P_{3(\alpha+\delta)})\rightarrow \homo_i(\manifold_{3(\alpha+\delta)})
$$
On the other hand, if 
$3(\alpha+\delta) < \frac{\cos 2\theta}{1+\cos 2\theta}$ 
and $\manifold_{3(\alpha+\delta)}\cap M_\theta=\emptyset$, we can use
Proposition~\ref{manifold-homotopy} and 
Lemma 3.2 in~\cite{CO08} to claim that 
$$
{\mathrm{image}\,}(\homo_i(\cup P_{\alpha+\delta})\rightarrow 
\homo_i(\cup P_{3(\alpha+\delta)})) \cong \homo_i(\manifold).
$$
Let $C^{\alpha}(P)$ denote the nerve of $\cup P_\alpha$.
One can recognize the resemblance between $C^{\alpha}(P)$ and
the well-known \v{C}ech complex. Both are nerves of unions of 
closed balls, but unlike \v{C}ech complexes, $C^{\alpha}(P)$
is the nerve of a union of balls that may have different radii;
recall that $\alpha$ denotes a fraction relative
to a distance rather than an absolute distance. 
The Nerve Lemma~\cite{borsuk48} provides that $C^{\alpha}(P)$ is
homotopy equivalent to $\cup P_{\alpha}$. Also, the argument
of Chazal and Oudot~\cite{CO08} to prove Theorem 3.5 can be extended to claim that for any $i \ge 0$, 
$$
{\mathrm{rank}}\, (\homo_i(C^{\alpha+\delta}(P))
\rightarrow \homo_i(C^{3(\alpha+\delta)}(P))= 
{\mathrm{rank}}\, 
(\homo_i(\cup P_{\alpha+\delta})\rightarrow 
\homo_i(\cup P_{3(\alpha+\delta)})) = {\mathrm{rank}}\,\homo_i(\manifold).
$$

The complex $C^{\alpha}(P)$ interleaves with another complex
$R^{\alpha}(P)$ that is reminiscent of the interleaving
of the \v{C}ech with the Vietoris-Rips complexes.
Specifically, let 
\[
R^{\alpha}(P):= \{\sigma\,|\, d(p,q)\le \alpha(d(p,L)+d(q,L))
\mbox{ for every edge $pq$ of $\sigma$} \}.
\]
It is easy to observe that $R^{\alpha}(P)$ is the completion of
the $1$-skeleton of $C^{\alpha}(P)$ and the following inclusions
hold as in the case of the original \v{C}ech and Vietoris-Rips complexes.
\[
C^{\alpha}(P)\subseteq R^{\alpha}(P) \subseteq C^{2\alpha}(P)
\mbox{ for any $\alpha\geq 0$}.
\]
Now, by choosing $\alpha + \delta \le \frac{1}{6} \frac{\cos 2\theta}{1+\cos 2\theta}$ (which also implies $\alpha + \delta \le \frac{1}{12}$ since $ \frac{\cos 2\theta}{1+\cos 2\theta}\le \frac{1}{2}$), we have a sequence similar to (\ref{eqn:C}) that eventually induces the following sequence: 
\[
\homo_i(C^{\alpha+\delta}(P))\rightarrow \homo_i(R^{\alpha+\delta}(P)) 
\rightarrow \homo_i(C^{2(\alpha+\delta)}(P)) \rightarrow \homo_i(C^{6(\alpha+\delta)}(P)) 
\rightarrow \homo_i(R^{6(\alpha+\delta)}(P))
\rightarrow \homo_i(C^{12(\alpha+\delta)}(P)). 
\]
In particular, following a similar argument as before, we have that 
$${\mathrm{rank}}\, (\homo_i(C^{\alpha+\delta}(P))
\rightarrow \homo_i(C^{12(\alpha+\delta)}(P))= 
{\mathrm{rank}}\, (\homo_i(C^{2(\alpha+\delta)}(P))
\rightarrow \homo_i(C^{6(\alpha+\delta)}(P)) = {\mathrm{rank}}\,\homo_i(\manifold)$$ 
as long as $12(\alpha+\delta)\le \frac{\cos 2\theta}{1+\cos 2\theta}$
and $\manifold_{12(\alpha+\delta)}\cap M_\theta=\emptyset$.
By using the standard results of interleaving~\cite{CO08} on this sequence, 
we obtain that
\begin{eqnarray*}
{\mathrm{rank}}(\homo_i(R^{\alpha+\delta}(P))\rightarrow
\homo_i(R^{6(\alpha+\delta)}(P)))
= {\mathrm{rank}}\,\homo_i(\manifold). 
\label{topo-eq}
\end{eqnarray*}
\vspace*{-0.1in}
\begin{theorem}
For a finite set $L\subset \mathbb{R}^k$ where
$L\cap \manifold=\emptyset$, let $P$ be a $(\delta, L)$-sample of 
the manifold $\manifold \subset \mathbb{R}^k$. 
Let $\theta \le \frac{\pi}{4}$, and 
$\alpha+ \delta \le \frac{1}{12}\frac{\cos 2\theta}{1+\cos 2\theta}$. 
If $\manifold_{12(\alpha+\delta)} \cap M_\theta = \emptyset$, then $\mathrm{rank} (\homo_i(\manifold)) = \mathrm{rank} (\homo_i(R^{\alpha+\delta}(P))\rightarrow \homo_i(R^{6(\alpha+\delta)}(P)))$, for any $i \ge 0$. 
\label{thm:topoinf-general}
\end{theorem}

\subsection{Computations for topology inference}
\begin{algorithm}[h!]
\floatname{algorithm}{Algorithm}
\caption{{{\sc LeanTopo}($P$)}}
\label{alg:leantopo}
\begin{algorithmic}[1]
        \STATE $\beta:=\frac{\pi}{5}$; ~~$\rho:= \frac{1}{26}\frac{\cos2\beta}{1+\cos2\beta}$; ~~$Q:={{\sc Lean}}(P,\beta,\rho)$; 
	\STATE Compute the complexes $R^{2\rho}(Q)$ and $R^{12\rho}(Q)$;
	\STATE Compute the persistence induced by the inclusion
$R^{2\rho}(Q) \rightarrow R^{12\rho}(Q)$. 
        \end{algorithmic}
\end{algorithm}
In step~3 of {\sc LeanTopo}, we compute the persistence
homology induced by the inclusion $R^{2\rho}(Q)\rightarrow R^{12\rho}(Q)$
where $\rho=\frac{1}{26}\frac{\cos2\beta}{1+\cos2\beta}$. 
When the parameter $\eps$ is sufficiently small and
$\beta=\frac{\pi}{5}$, we can find a value 
$\theta$ such that $\frac{\pi}{4}\ge \theta \ge \beta + \frac{3}{2}\sqrt{\eps} + \nu_\eps$ and
$2\rho = \frac{1}{13}\frac{\cos2\beta}{1+\cos2\beta} \le \frac{1}{12} \frac{\cos2\theta}{1+\cos2\theta}$. 
This is precisely what is  needed for the homology inference 
in Theorem \ref{thm:topoinf-general}. 
More specifically, recall 
by Eqn.~\ref{eqn12} in the proof of 
Theorem \ref{thm:sparsification}, the output sparsified set of points $Q$ 
is a $(\delta, L_{\frac{\pi}{5}})$-sample for
$\delta=\frac{6}{5}\rho$.
The algorithm implicitly sets $\alpha = 2\rho - \delta = \frac{4}{5}\rho$ such that $\alpha + \delta = 2\rho \le \frac{1}{12} \frac{\cos2\theta}{1+\cos2\theta}$ when $\eps$ is sufficiently small. 
Theorem~\ref{thm:topoinf-general} requires further that the offset
$\manifold_{\alpha'}:= h^{-1}_{L_{\frac{\pi}{5}}}[0,\alpha']$
is disjoint from $\thetaM$ for $\alpha'= 12(\alpha+\delta)$ which
we establish using the following proposition.
 
\begin{proposition}
Let $\alpha'\leq \frac{1}{1+\cos\theta+\eps}$ and $\theta$ be such that 
$\frac{\pi}{2}\geq \theta\geq \frac{\pi}{5}+\frac{3}{2}\sqrt\eps+\nu_\eps$ for
a sufficiently small $\eps\leq \frac{1}{8}\sin^2\theta$.
Then, 
$
\thetaM\cap \manifold_{\alpha'}=\emptyset.
$
\label{avoidM}
\end{proposition}

\begin{proof}
We prove the result by contradiction. Assume that there exists a point $x\in M_\theta\cap  \manifold_{\alpha'}$. 
Define $m$ and $s$ as in the proof of Proposition \ref{med-approx}. 
With $\beta=\frac{\pi}{5}$,
the assumed conditions for $\theta,\beta,\eps$ are same
as in Proposition~\ref{med-approx}, and thus we can arrive at the
inequality~\ref{eqn:A} in its proof.
Since $s$ is a closest point of $m$ in $\manifold$,
we have $d(s,m) = d(m, \manifold) \leq d(x,m)+ d(x,\manifold)$.
Combining this with Eqn (\ref{eqn:A}), it follows that, 
for any $x\in \manifold_{\alpha'}$,
$$
d(x,L_{\frac{\pi}{5}})\leq
(1+\cos\theta+\eps) \cdot d(x, m)+ (\cos\theta+\eps) \cdot d(x,\manifold).
$$
Since $x \in \manifold_{\alpha'}$, $h_{L_{\frac{\pi}{5}}}(x) \le \alpha'$ implies
that $d(x,\manifold)\leq \frac{\alpha'}{1-\alpha'}d(x,L_{\frac{\pi}{5}})$.
Hence $d(x, L_{\frac{\pi}{5}}) \le c \cdot d(x, m) = c \cdot d(x, M_\theta)$ for the positive constant $c = \frac{1+\cos \theta + \eps}{1 - \frac{\alpha'}{1-\alpha'} (\cos \theta + \eps)} = 1 + \frac{\cos \theta + \eps}{1 - \alpha'(1 +\cos \theta + \eps)}$.

On the other hand, since $x\in M_\theta\cap  \manifold_{\alpha'}$ and 
since $\manifold \cap M_\theta$ is empty, $x\not\in \manifold$.
Thus, $d(x,\manifold)>0$. 
Since $x \in \manifold_{\alpha'}$
and $h_{L_{\frac{\pi}{5}}} (x) \le \alpha'$, we have
that $d(x,L_{\frac{\pi}{5}})\geq \frac{1-\alpha'}{\alpha'} d(x, \manifold)$.
Hence $d(x, L_{\frac{\pi}{5}}) > 0$ as well since $\alpha' < 1$.
This further implies that $d(x,M_\theta)>0$ because
according to the above derivation, 
$d(x,M_\theta) \ge \frac{1}{c} d(x,L_{\frac{\pi}{5}})$ for $c>0$.
This however contradicts the fact that $x\in M_\theta \cap \manifold_{\alpha'} \in M_\theta$. 
Hence our assumption is wrong and there is no such point $x\in M_\theta\cap  \manifold_{\alpha'}$. 
\end{proof}

\begin{theorem}
Let $\manifold\subset \mathbb{R}^k$ be a smooth compact manifold 
without boundary of known intrinsic dimension.
Let $P$ be an $\eps$-dense sample of $\manifold$ w.r.t. $\lfs$. 
{\sc LeanTopo}$(P)$ computes the 
rank of $\homo_i(\manifold)$ for any $i \ge 0$ when $\eps$ 
is sufficiently small. 
\label{thm:topoinfo-alg}
\end{theorem}
\begin{proof}
Since $\frac{\cos 2\theta}{1+\cos 2\theta}\leq \frac{1}{1+\cos \theta+\eps}$
for $\theta\leq \frac{\pi}{2}$ and small enough $\eps$, 
one has the fact that $\alpha'= 12 (\alpha+\delta) \le \frac{\cos 2\theta}{1+\cos 2\theta}$
implies that $\alpha'\leq \frac{1}{1+\cos \theta +\eps}$.
This means that the parameters $\alpha,$ and $\theta$ set
by the algorithm {\sc LeanTopo} implicitly or explicitly
satisfy the conditions required by Proposition 
\ref{avoidM}. Hence, $\manifold_{12(\alpha+\delta)} \cap M_\theta = \emptyset$. 
Therefore, all conditions for Theorem \ref{thm:topoinf-general} hold
for the sparsified set $Q$ output by {\sc Lean}, and it then follows that $\mathrm{rank}(\homo_i(\manifold)) = \mathrm{rank}(\homo_i(R^{2\rho}(Q)) \to \homo_i(R^{12 \rho}(Q)))$, for any $i\ge 0$. 
\end{proof} 

We remark that a particular interesting feature of 
Algorithm {\sc LeanTopo} is that, we only need to 
set the parameter $\beta$ to a universal 
constant $\frac{\pi}{5}$. All other
parameters such as the angle and radius conditions for
choosing $\beta$-good pairs and the decimation radius 
are determined by this choice of the angle $\beta$. 
This makes {\sc LeanTopo} parameter-free; 
see also our experimental results in Section~\ref{sec:experiment}.
At the same time, the above Theorem states that its output is guaranteed to be correct as the input set of samples $P$ becomes sufficiently dense.

\section{Experiments and discussion}
\label{sec:experiment}
We experimented with {\sc LeanTopo} primarily on curve and surface samples.
We used thresholds for sparsification that are more aggressive than
predicted by our analysis. For example, our analysis predicts that
for $\beta=\frac{\pi}{5}$, the constant
$c_\beta=\frac{1}{3}\tan \frac{\beta}{2}\approx 0.11$, but
we kept it at $0.5$.
We kept the same thresholds for all
models to ensure that we don't fine tune it for different input.
The decimation ratio $\frac{r_q}{d(q,L_\beta)}$ is kept at $0.5$,
and the $r$ for
computing the complex $R^r$ is kept at $0.65$ in all cases.
Table~\ref{tbl:experiment} below shows the details. The rank of $H_1$ homology is computed 
correctly by our algorithm for all these data. The sparsified points are shown in Figure \ref{experiment}.

\begin{table*}[htbp]
\begin{center}
\begin{tabular}{|c|c|c|c|c|c|c|}
\hline
Name & input \#points & output \#points & $c_\beta$ & decimation ratio & $r$ for $R^r$ & ${\mathrm rank}\, H_1$ \\
\hline
{\sc MotherChild} & 126500 & 5267 & 0.5 & 0.5 & 0.7 & 8 \\
\hline
{\sc Botijo} & 101529 & 7600 & 0.5 & 0.5 & 0.7 & 10 \\
\hline
{\sc Kitten} &  134448 & 1914 & 0.5 & 0.5 & 0.7 & 2 \\
\hline
{\sc CurveHelix} & 1000 & 235 & 0.5 & 0.5 & 0.7 & 1 \\
\hline
\end{tabular}
\end{center}
\caption{Experiments on a curve and three surface samples.}
\label{tbl:experiment}
\end{table*}

\paragraph*{Extensions.}
One obvious question that remains open is how to
extend the scope of our sparsification strategy to
larger class of input, such as noisy data samples and/or samples
from compact spaces rather than manifolds. 

{\bf Noise:} We observe that, for Hausdorff noise, where
samples are assumed to lie within a small offset of the manifold, 
our method can be applied.
However, a parameter giving the extent
of this Hausdorff noise needs to be supplied. With this
parameter, one can estimate the normals reliably from the noisy
but dense sample. The step where
we compute the lean set, requires an empty ball test which
also needs this parameter because otherwise noise can collaborate
to provide a false impression that some spurious manifolds have been sampled.
Given the ambiguity that a noisy sample can be dense for
two topologically different spaces, it may be impossible to avoid
a parameter that eliminates different such possibilities. Nevertheless,
our method would free the user from specifying a threshold for
building the complexes.

In an experiment, we added artificial noise on the three surface samples 
as shown in Figure~\ref{fig:noise} to test robustness of our algorithm. 
We added a uniform displacement to each sample point along
the normal direction. The displacement ranged from $-0.5\%$ to $0.5\%$ times 
the diameter of the model. We modified our algorithm to ignore all leanset 
points formed by two points closer than a threshold which is picked as a
multiple of the diameter of the model. Other thresholds were kept the same 
as in the previous experiment. Results in Table~\ref{tbl:noise-scale}
show that the algorithm can tolerate 
noise in case there is a known upper limit on the noise level.
\begin{table*}
        \begin{center}
                \begin{tabular}{|p{3.5cm}|c|c|c|c|c|c|c|c|c|c|c|c|c|c|}
                        \hline
                        {Threshold (multiple of noise scale)}&  0 & 1 & 2 & 3 & 4 & 5 & 6 & 7 & 8 & 9 & 10 & 11 & 12 & 13 \\
                        \hline
                        {\sc MotherChild} &
                        18196 & 1636 & 37 & 8 & 8 & 8 & 8 & 8 & 8 & 8 & 8 & 8 &
                        7 & 7\\ \hline
                        {\sc Botijo} &
                        14565 & 14580 & 1462 & 10 & 10 & 10 & 10 & 10 & 10 & 10 & 10 & 8 & 8 & 8 \\ \hline
                        {\sc Kitten} &
                        20506 & 20572 & 1314 & 2 & 2 & 2 & 2 & 2 & 2 & 2 & 2 & 2 & 2 & 2 \\ \hline
                \end{tabular}
        \end{center}
        \caption{Experiments on 3 surfaces with artificial noise. The table shows resulting $H_1$ of each model under different threshold. Experiments show that the influence of noise is removed when we pick threshold greater than or equal to 3 times of the noise scale. The threshold might introduce problem when it is too large. }
        \label{tbl:noise-scale}
\end{table*}

\begin{figure*}[ht!]
        \begin{center}
        \includegraphics[width=0.95\textwidth]{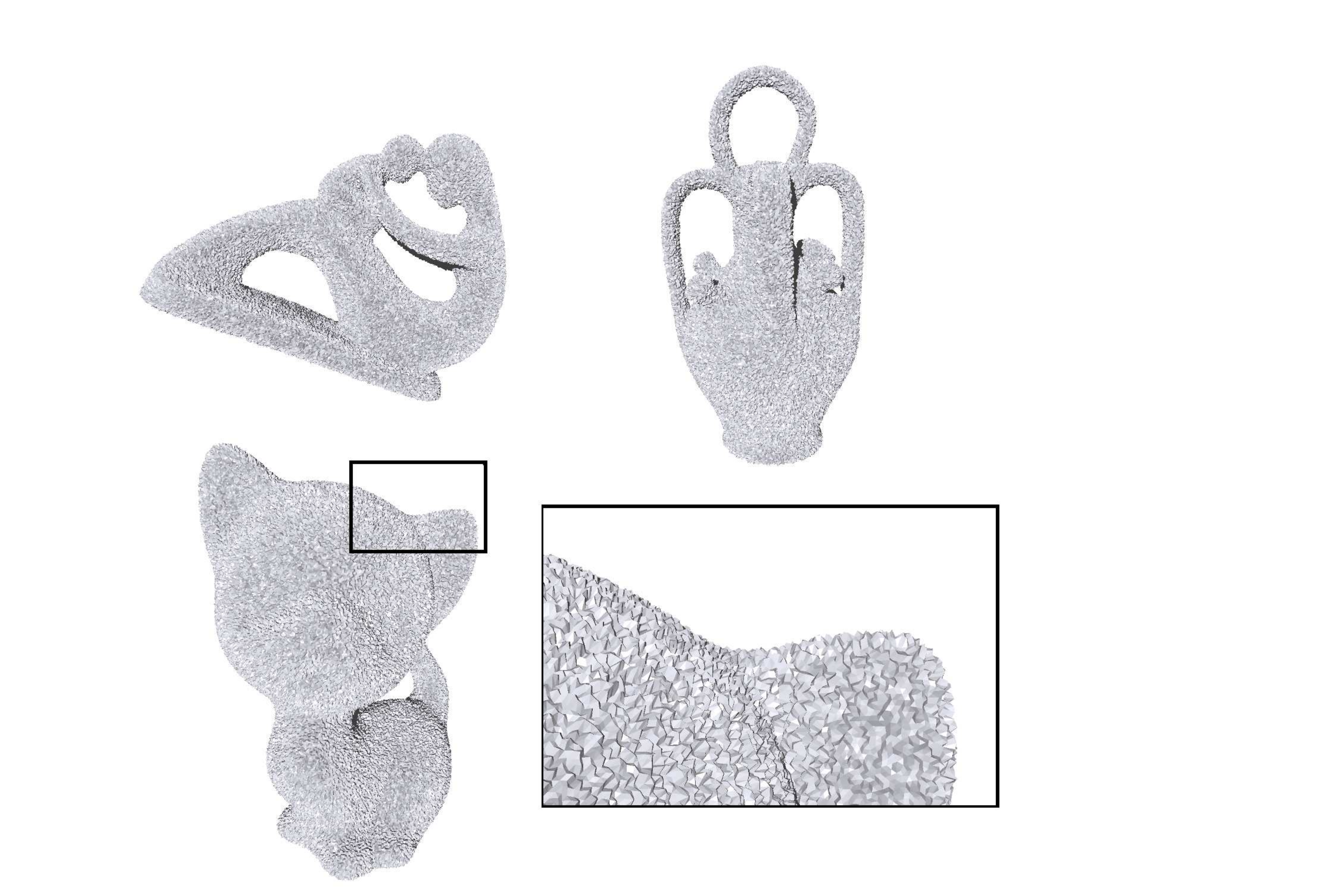}
        \end{center}
        \caption{Noisy samples. Meshes are created only for rendering.
        }
        \label{fig:noise}
\end{figure*}

The more general noise model which allows outliers would also be
worthwhile to investigate.
One may explore the `distance to measure'
technique proposed in~\cite{CCM11} for this case. But, it is not clear
how to adapt the entire development in this paper to this setting.
One possibility is to eliminate all outliers first to make the
noise only Hausdorff, and then apply the technique for Hausdorff
noise as alluded in the previous paragraph.
This will certainly require more parameters to be supplied by the user. 

{\bf Compacts:}
The case for compact sets is perhaps more challenging. The normal spaces
are not well defined everywhere for such spaces. Thus, we need to devise
a different strategy to compute the lean set. The theory of compacts
developed in the context of topology inference in~\cite{CCL09} may be
useful here. Computing the lean sets efficiently
in high dimensions for compact spaces remain a formidable open problem. 

\section*{Acknowledgment}
This work was partially supported by the NSF grants
CCF-1064416, CCF-1116258, CCF 1318595, and CCF 1319406.

\newpage 

\appendix
\section{Missing Proofs}
\label{appendix:A}
\paragraph{Proving that $(p,q)$ is $\beta$-good for 
Proposition~\ref{med-approx}.}
We know that $\angle{(\nor_s, st)} \leq \pi/2-\theta$ which implies 
that $d(s,t)\geq 2d(s,m)\sin\theta\geq 2\lfs(s)\sin\theta$.
Consider the triangle $pst$. By triangle inequality, $d(p, t) \ge d(s, t) - d(s, p) \ge (2\sin \theta - \eps) \lfs(s)$.
The angle $\angle{pts}$ is at most
\begin{equation}
\arcsin \frac{d(p,s)}{d(p,t)}\leq \arcsin \frac{\eps \lfs(s)}{(2\sin\theta - \eps) \lfs(s)} \leq \frac{4}{3} \cdot \frac{\eps}{2\sin\theta - \eps}.
\label{eqn:smallangle1}
\end{equation}
The last inequality follows from that
$\arcsin (x) \le cx$ for $x \le \frac{\sqrt{c^2-1}}{c}$. In our case, choose $c = \frac{4}{3}$. Since $\sqrt{\eps} \le \frac{\sin \theta}{2\sqrt{2}} \le \frac{1}{2}$, we have that
$$\frac{\eps}{2\sin \theta - \eps} \le  \frac{\eps}{4 \sqrt{\eps} - \eps} = \frac{\sqrt{\eps}}{4 - \sqrt{\eps}} \le  \frac{1}{7} \le \frac{\sqrt{c^2-1}}{c}. $$
Now assume without loss of generality that
$\lfs(s) \geq \lfs(t)$. Then,
$$d(p,q) \geq d(s,t)- d(p,s)-d(q,t) \geq d(s,t)-2\eps \lfs(s)\geq
2(\sin\theta-\eps)\lfs(s).
$$
Recall that $d(t, p) \ge (2\sin \theta - \eps)\lfs(s)$. Considering the triangle $tpq$, we have
\begin{equation}
\angle{tpq}\leq \arcsin\frac{d(q,t)}{d(p,t)}\leq
\arcsin\frac{\eps \lfs(t)}{2(\sin\theta-\eps)\lfs(s)}
\leq \arcsin\frac{\eps \lfs(s)}{2(\sin\theta-\eps)\lfs(s)}
\leq \frac{4}{3} \cdot \frac{\eps}{2(\sin\theta-\eps)},
\label{eqn:smallangle2}
\end{equation}
where the last inequality follows from a similar argument used for Eqn. (\ref{eqn:smallangle1}).

We know that,
$
\angle{(\nor_p,\nor_s)}\leq \eps
$,
$
\angle{(\vv_p,\nor_p)}\leq \nu_{\eps},
$
and $\angle{(pq, st)} \le \angle pts + \angle tpq$.
Combining these with Eqn. (\ref{eqn:smallangle1}), (\ref{eqn:smallangle2}) and the assumption that $\sqrt{\eps}\le \frac{1}{2\sqrt{2}}\sin \theta (\le \frac{1}{2})$, we have that
\begin{align*}
\angle{(pq,\vv_p)} &\leq \angle{(pq,st)}+ \angle{(st,\nor_s)}+\angle{(\nor_s,\nor_p)} +
 \angle{(\nor_p,\vv_p)}
\leq \frac{8}{3}\cdot \frac{\eps}{2\sin\theta - 2\eps}+\frac{\pi}{2}-\theta+\eps +\nu_\eps \\
&\leq \frac{8}{3} \cdot \frac{\sqrt{\eps}}{4\sqrt{2} - 2 \sqrt{\eps}} + \frac{\sqrt{\eps}}{2} + \frac{\pi}{2} - \theta + \nu_\eps \le \frac{\pi}{2} - \theta + \frac{3}{2} \sqrt{\eps} + \nu_\eps
\end{align*}
Similar bound holds for $\angle{(pq,\vv_q)}$.
It follows that the pair $(p,q)$ satisfies the
first condition of being $\beta$-good, as long as $\frac{\pi}{2}-\theta +
\frac{3}{2}\sqrt{\eps} +\nu_\eps\leq \frac{\pi}{2}-\bb$. This is guaranteed by requiring $\theta\geq \bb + \frac{3}{2} \sqrt\eps +\nu_\eps$ (as specified in the proposition).


\parpic[r]{\includegraphics[height=4cm]{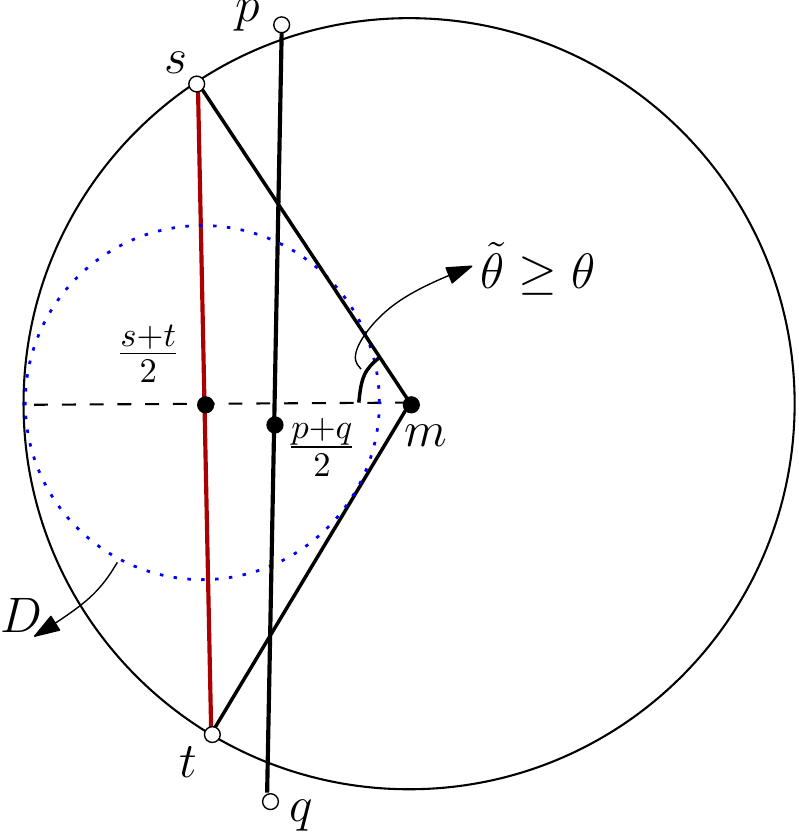}}
Next, we argue that $(p,q)$ also satisfies the second condition of being
$\beta$-good.
To do so, let $\newtheta = \frac{1}{2} \angle smt$ be half of the angle spanned by $sm$ and $tm$. Note that by the definition of $\theta$-medial axis $M_\theta$, we have that $\newtheta \ge \theta$.
See the right figure for an illustration.
First, observe that the ball $D = B(\frac{s+t}{2}, r)$ with $r = d(s,m) \cdot (1-\cos \newtheta)$ does not intersect $\manifold$, since this ball is contained inside the medial ball $B(m, d(s,m))$.
The midpoint $\frac{p+q}{2}$ of $pq$
is at most $\eps \lfs(s) \le \eps d(s,m)$ distance away from
$\frac{s+t}{2}$ because both $p$ and $q$ are at most $\eps \lfs(s)$
away from $s$ and $t$ (assuming w.o.l.g $\lfs(s)\geq \lfs(t)$).
This means that the ball $D'=B(\frac{p+q}{2},r')$ centering at the midpoint of $pq$ and with radius
$r'=d(s,m) \cdot (1-\cos\newtheta -\eps)$ is contained in the ball $D$ and thus
does not have any point of $\manifold$ and hence $P$ inside.

On the other hand, note that
$$d(p, q) \le d(s,t) + 2\eps \lfs(s) \le 2 d(s,m) \sin \newtheta + 2\eps d(s,m) = 2d(s,m) (\sin \newtheta + \eps). $$
Thus,
the second condition for
$p,q$ being a good pair is satisfied as long as
$$
c_\bb\leq\frac{1-\cos\newtheta-\eps}{2(\sin\newtheta+ \eps)} \leq \frac{r'}{d(p,q)}.
$$
Consider the function $f(x) = \frac{1-\cos x-\eps}{\sin x+ \eps}$, its derivative $f'(x)$ is greater than $0$ for $x \in [0, \pi/2]$. Indeed,
$$
f'(x) = \frac{\sin x \cdot (\sin x + \eps) - (1- \cos x - \eps)\cdot \cos x}{(\sin x + \eps)^2} = \frac{1 - \cos x + \eps \sin x + \eps \cos x}{(\sin x + \eps)^2} \ge 0.
$$
Hence $f(x)$ is an increasing function, and $f(\newtheta) \ge f(\theta)$ since $\newtheta \ge \theta$. In other words, the second condition for $(p,q)$ being a good pair is satisfied as long as $c_\bb \le \frac{1-\cos\theta-\eps}{2(\sin\theta+ \eps)}$.
To further simplify it, note that using $\eps \le \frac{1}{8}\sin^2 \theta$, one can show that $\frac{4\eps}{\sin \theta} \le  \tan \frac{\theta}{2}$.
Combining this with $\frac{1-\cos \theta}{\sin \theta} = \tan \frac{\theta}{2}$, we then have $$
\frac{1-\cos\theta-\eps}{2(\sin\theta+\eps)}\geq
\frac{1-\cos\theta-\eps}{\frac{9}{4}\sin\theta}=\frac{4}{9}\tan\frac{\theta}{2}
-\frac{4\eps}{9\sin\theta} \geq \frac{4}{9}\tan\frac{\theta}{2} - \frac{1}{9} \tan \frac{\theta}{2} = \frac{1}{3} \tan\frac{\theta}{2} \geq \frac{1}{3} \tan\frac{\bb}{2}.
$$
Hence as $c_\bb \le\frac{1}{3}\tan\frac{\bb}{2}$, the ball $B(\frac{p+q}{2}, c_\bb d(p,q))$ is contained in $D'$ and thus contains no point in $P$.
Therefore, the pair $(p,q)$ is $\beta$-good and its midpoint is in $L_\bb$.

\paragraph{Proof of Theorem~\ref{thm:sparsification}.}
Let $x$ be any point in $\manifold$ to which $p$ is the nearest
sample point in $P$. Then, $d(x,p)\leq \eps \lfs(x)\leq \eps'\lfs(p)$
where $\eps'=\frac{\eps}{1-\eps}$.
If $p$ is retained in $Q$, $d(x,Q)\leq \eps \lfs(x)\le\eps'\lfs(p) \le \frac{\eps'}{\cone} d(p, L_\bb) \le \frac{6\rho}{5} d(p, L_\bb)$ for sufficiently small $\eps > 0$, where $\cone$ is the constant from Proposition \ref{lb-to-med}.
Now consider the case when $p$ is deleted while processing another point, say $q \in P$.
By the decimation procedure in lines 5--9, $d(q, L_\bb) \ge d(p, L_\bb)$ and $q$ will remain in $Q$ since we process points in non-decreasing
order of their distances to $L_\beta$.
Using Proposition \ref{lb-to-med}, we then have:
\begin{align*}
d(x, q) &\le d(x,p) + d(p, q) \le \eps \lfs(x) + \rho d(q, L_\bb) \le \eps' \lfs(p) + \rho d(q, L_\bb) \\
&\le \frac{\eps'}{\cone} d(p, L_\bb) + \rho d(q, L_\bb) \le (\frac{\eps'}{\cone} + \rho) d(q, L_\bb) \le \frac{6\rho}{5} d(q, L_\bb).
\end{align*}
The last inequality holds when $\eps$ is sufficiently small (in which case the estimation error $\nu_\eps$ in the normal space is also small).
Therefore,
\begin{eqnarray}
\label{eqn12}
d(x,q) \leq \frac{6\rho}{5}\lnfs(q)
\end{eqnarray}
Now applying Remark~\ref{lnfs-remark}, $Q$ is also $(\frac{4}{3}\rho)$-dense
because $\frac{\frac{6\rho}{5}}{1-\frac{6\rho}{5}}\leq \frac{4}{3}\rho$
for $\rho\le \frac{1}{12}$.

The fact that $Q$ is $\rho$-sparse w.r.t. $\lnfs$ follows easily from the 
decimation procedure.

%
%

\section{Estimation of Normal/Tangent Space}
\label{appendix:normal}

\newcommand{\newsigma}	{\widehat{\sigma}}
\newcommand{\myDelta}		{D}
\newcommand{\myvol}		{\mathrm{Vol}}
\newcommand{\newphi}		{\beta}
   
Here, we provide the justification for the claimed bound of $O(\eps)$
on the tangent space estimation(and thus the normal space) of the 
hidden manifold $\manifold$ at a sample point $p \in P$. 
For completion, we restate the procedure described in
section~\ref{sec:computation} for estimating the tangent space $T_p$. 
Set $\beta=\frac{\pi}{5}$ for the calculations to follow.
Let $\inD$ denote the intrinsic dimension of the manifold $\manifold$, which we assume is known a-priori. 
Let $p_1 \in P$ be the nearest neighbor of $p$ in $P \setminus \{p\}$. 
Suppose we have already obtained points $\sigma_i = \{p, p_1, \ldots, p_i\}$ with $i < \inD$. 
Let $\aff(\sigma_i)$ denote the affine hull of the points in $\sigma_i$. 
Next, we choose $p_{i+1}\in P$ that is closest to $p$ among all points
forming an angle within the range $[\frac{\pi}{2} - \bb, \frac{\pi}{2}]$ 
with $\aff(\sigma_i)$. 
We add $p_{i+1}$ to the set and obtain $\sigma_{i+1} = \{p, p_1, \ldots, p_i, p_{i+1}\}$. 
This process is repeated until $i = \inD$, at which point we have obtained $\inD+1$ points $\sigma_\inD = \{p, p_1, \ldots, p_{\inD}\}$. 
We use $\aff(\sigma_\inD)$ to approximate the tangent space $T_p$. 
We now show that the simplex $\sigma_i$ is ``fat''. In particular, we will leverage a result (Corollary 2.6) of \cite{BG11} to bound the angle between the true tangent space $T_p$ and approximate tangent space $\aff(\sigma_i)$. 

More specifically, we first modify the simplex $\sigma_i$ to another one $\newsigma_i$ as follows. 
Let $\myDelta$ denote the longest length of any edge incident to $p$ 
in $\sigma_i$. 
Later we will prove that $\myDelta = O(\frac{\eps \lfs(p)}{\sin \newphi})$. 
Now, we extend each edge $pp_j$ along the same line segment but to $p\hat{p}_j$ such that $\| p \hat{p}_j \| = \myDelta$. 
The resulting simplex spanned by $\{p, \hat{p}_1, \ldots, \hat{p}_i \}$ is denoted by $\newsigma_i$. 
By construction, $\aff(\sigma_i) = \aff(\newsigma_i)$. 
Hence, we only need to bound the angle $\angle (T_p, \aff(\newsigma_i) )$. 
Corollary 2.6 of \cite{BG11} states that $\sin \angle( \aff(\newsigma_i), T_p) \le \frac{L^{i+2}}{\myvol(\newsigma_i) \cdot S \cdot \lfs(p)}$, where $L$ and $S$ are the longest and shortest edge length of $\newsigma_i$ respectively; while $\myvol(\newsigma_i)$ stands for the volume of the simplex  $\newsigma_i$. 
To use this result, we bound the terms $L, S$, and $\myvol(\newsigma_i)$. 

\parpic[r]{\includegraphics[height=3.5cm]{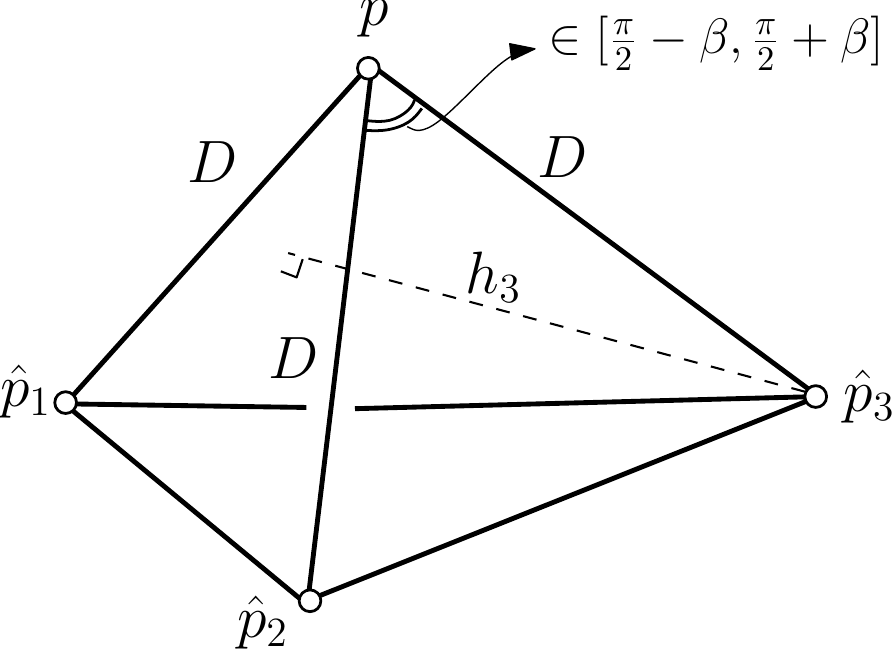}}
See the figure on right for an illustration. 
First, we bound the angle between any two $p \hat{p}_{\ell}$ and $p\hat{p}_j$, for $\ell, j \in [1, i]$. 
Assume w.o.l.g. that $j > \ell$. 
By construction, $p \hat p_j$ forms an angle $\alpha$ such that $\alpha \in [\frac{\pi}{2} - \newphi, \frac{\pi}{2}]$ with $\aff(\sigma_{j-1})$. 
It follows that $\alpha \le \angle (p \hat p_{\ell}, p\hat p_j) \le \pi - \alpha$, that is, $\angle (p \hat p_{\ell}, p\hat p_j) \in  [\frac{\pi}{2} - \newphi, \frac{\pi}{2} + \newphi]$. 
Therefore, the edge length $d(\hat p_{\ell}, \hat p_j)$ satisfies 
$$d(\hat p_{\ell}, \hat p_j ) = 2\myDelta \cdot \sin \frac{1}{2} \angle (p \hat p_{\ell}, p\hat p_j) \in [2\myDelta \cdot \sin(\frac{\pi}{4} - \frac{\newphi}{2}) , 2\myDelta \cdot \sin(\frac{\pi}{4} + \frac{\newphi}{2})]. $$
Therefore the longest edge length $L$ in simplex $\newsigma_i$ is at most $L \le 2\myDelta \cdot \sin(\frac{\pi}{4} + \frac{\newphi}{2})$, while the smallest edge length $S$ in simplex $\newsigma_i$ is at least $S \ge \min\{ D, 2\myDelta \cdot \sin(\frac{\pi}{4} - \frac{\newphi}{2}) \}$. 

Next, we bound the volume $\myvol(\newsigma_i)$ of $\newsigma_i$, which we do inductively. 
We claim that $\myvol(\newsigma_i) \ge \frac{D\cdot (D\cdot \cos \newphi)^{i-1}}{i !}$. 
This claim holds when $ i = 1$ in 
which case $\myvol(\newsigma_1) = d(p,\hat p_1) = D$. 
Assume it holds for $i-1$. 
Then, we have that $\myvol(\newsigma_i) = \frac{1}{i} d(\hat p_i, \aff(\newsigma_{i-1})) \cdot \myvol(\newsigma_{i-1})$, where $h_i = d(\hat p_i, \aff(\newsigma_{i-1}))$ is the height of the simplex $\newsigma_i$ using $\newsigma_{i-1}$ as the base facet.
On the other hand, by construction $\angle(p \hat p_i, \aff(\newsigma_{i-1}) \ge \frac{\pi}{2} - \newphi$, which gives
$$h_i = d(p, \hat p_i)  \cdot \sin \angle(p \hat p_i, \aff(\newsigma_{i-1})) \ge \myDelta \cdot \cos \newphi. $$
It follows that $\myvol(\newsigma_i) \ge \frac{\myDelta}{i} \cdot \cos \newphi \cdot \myvol(\newsigma_{i-1}) \ge \frac{ D\cdot (D\cdot \cos \newphi)^{i-1}}{i !} $, which then proves the claim inductively.

Now we derive an upper bound on $D$. Inductively, assume that for $1\leq i \leq s$,
$$
D\leq 13\eps\lfs(p) \mbox{ and } 
\theta_i=\angle{(\aff(\sigma_i),T_p)}\leq \arcsin\left(\frac{i!2^{i+2} D \sin^{i+2}(\frac{\pi}{4}+\frac{\beta}{2})}{\cos^{i-1}\beta \sin(\frac{\pi}{4}-\frac{\beta}{2})\lfs(p)}\right).
$$ 
For $i=1$ and sufficiently small $\eps$, it is true because the nearest 
point $p_1$ to $p$ satisfies
$d(p,p_1)\leq 3\eps\lfs(p)$
and also $\sin\angle({pp_1,T_p})\leq \frac{3}{2}\eps$~(this follows 
easily from the $\eps$-dense sampling condition, see e.g. Corollary 3.1 and
Lemma 3.4~\cite{Dey07}) 
For induction consider
the time when we choose $p_i$. 
Consider the projection $\tilde{\sigma}_{i-1}$ of $\sigma_{i-1}$
onto $T_p$ and the $(i-1)$-dimensional affine subspace 
$\aff(\tilde{\sigma}_{i-1})$ of $T_p$ containing this projection.
By our inductive hypothesis, 
$\angle(\aff{\sigma_{i-1}},\aff{\tilde{\sigma}_{i-1}})\leq \theta_{i-1}$.
Let $F$ be the subspace of $T_p$ orthogonal to $\aff{\tilde{\sigma}_{i-1}}$
and let $x\in F$ be such that
$d(x,p)=10\eps\lfs(p)$. The closest point $\tilde{x}\in\manifold$ of $x$
to $\manifold$ has $d(x,\tilde x)=O(\eps^2\lfs(p))$. Therefore,
we can assume that 
$$9\eps\lfs(p)\leq d(p,\tilde x)\leq 11\eps\lfs(p)$$ when
$\eps>0$ is sufficiently small. There is a sample point $p'\in P$
with $d(\tilde x, p')\leq \eps \lfs(\tilde x)$. This means that
the angle $\angle{p\tilde x, pp'}$ is at most 
$\arcsin(\frac{\eps\lfs(x)}{9\eps\lfs(p)})=\arcsin \frac{1}{8}$ when
$\eps$ is sufficiently small. It follows that  
$\angle{(pp',\aff(\sigma_{i-1}))}\geq \frac{\pi}{2}-\arcsin\frac{1}{8}
-\theta_{i-1}$.
One can make $\theta_{i-1}$ arbitrarily small by choosing
$\eps$ sufficiently small. Therefore,
if $\beta=\frac{\pi}{5}$ and $\eps$ is small enough, we have 
$\angle{pp',\aff(\sigma_{i-1})}\in [\frac{\pi}{2}-\beta,\frac{\pi}{2}]$.
Since $p_i$ is chosen with the smallest distance from $p$ satisfying
the above angle condition, we have, for small enough $\eps$, 
$$d(p,p_i)\leq d(p,p')\leq d(p,\tilde x)+d(\tilde x,p')\leq 11\eps\lfs(p)+\eps\lfs(x)\leq 13\eps\lfs(p).$$
Since $D$ cannot be larger than the maximum between older $D$
from stage $i-1$ and $d(p,p_i)$, one has $D\leq 13\eps\lfs(p)$.
Combining all these with Corollary 2.6 of \cite{BG11}, we obtain that 
$\sin \angle(\aff(\newsigma_i), T_p)= \sin\theta_i$ as claimed. 

Evaluating $\sin\theta_i$ we obtain $\sin\theta_i=O(\frac{D}{\lfs(p)})=O(\eps)$
for all $i\in[1,s]$ 
where the big-O notation hides constants depending exponentially on the intrinsic dimension $\inD$ and $\cos\newphi$. 
In other words, the angle $\nu_\eps$ between the approximate tangent space and the true tangent space (thus between the approximate normal space and the true normal space) at any sample point is bounded by $O(\eps)$, where the big-O notations hides constant depending on the angle $\beta$ and intrinsic dimension $\inD$ of the manifold $\manifold$. 

\end{document}